\documentclass[twocolumn]{IEEEtran}
\pdfoutput=1
\usepackage{url}
\usepackage{amsmath}
 
\usepackage{amsthm}
\usepackage[linesnumbered,ruled]{algorithm2e}
\usepackage{verbatim}
\usepackage{graphicx}
\usepackage{subfigure}
\newtheorem{theorem}{Theorem}
\newtheorem*{theorem_nn}{Theorem}
\newtheorem{lemma}{Lemma}

\newtheorem{corollary}{Corollary}
\newtheorem{proposition}{Proposition}

\newtheorem{prob}{Problem}

\theoremstyle{definition}
\newtheorem{IP}{Integer Program}

\graphicspath{{figures/}}
\usepackage{amsmath}
\usepackage{mathrsfs}
\usepackage{amssymb}
\usepackage{booktabs}
\usepackage{changepage}

\DeclareMathOperator*{\argmax}{arg\,max}

\usepackage{tikz}
\usetikzlibrary{arrows}

\begin{document}
\title{Pseudo-Separation for Assessment
of Structural Vulnerability of a Network}
%%%VERSION FOR ARXIV
\author{Alan Kuhnle, Tianyi Pan, Victoria G. Crawford, Md Abdul Alim, and My T. Thai\\
Department of Computer \& Information Science \& Engineering\\
University of Florida\\
Gainesville, Florida, USA\\
Email: \{kuhnle, tianyi, crawford, alim, mythai\}@cise.ufl.edu }
\maketitle
\begin{abstract}
Based upon the idea that network functionality 
is impaired if two nodes in a network are 
sufficiently
separated in terms of a given metric,
we introduce two combinatorial \emph{pseudocut} 
problems generalizing the classical 
min-cut and multi-cut problems. 
We expect the pseudocut problems 
will find broad relevance to the 
study of network reliability.
We comprehensively analyze 
the computational complexity of the pseudocut problems 
and provide three approximation algorithms for these problems. 

Motivated by applications in communication networks
with strict Quality-of-Service (QoS) requirements,
we demonstrate the utility of the pseudocut problems by proposing
a targeted vulnerability assessment for the structure of  
communication networks using QoS metrics; 
we perform experimental evaluations 
of our proposed approximation algorithms in this context.
\end{abstract}
\section{Introduction}
%a. Classical cutting problems are important for a variety of applications and are well-studied
The concept of connectivity, or the existence of a path between two nodes, is 
vital for any network. Whatever functionality a network may provide to a pair
of nodes is usually absent if the pair is disconnected. As a result, many studies
of network vulnerability, or the degree to which the functionality of a network
may be disrupted by failures, have incorporated connectivity as a fundamental
measure of network functionality \cite{Grubesic2008,Arulselvan2009,Dinh2012b,Dinh2015a}.
Recognition of the importance of connectivity has led to the study of 
many combinatorial problems related to connectivity \cite{Leighton1999,Vazirani2001,Colbourn1987}, 
perhaps the most well-known of which is
the minimum cut problem (CUT),
of determining the minimum number of edges (vertices) to remove in order to disconnect
a pair $(s,t)$ of vertices in a graph. CUT was shown to be solvable in polynomial time
via the celebrated maximum flow minimum cut relationship \cite{Ford1956}.  

%b. However, instead of completely disconnecting the graph, it is interesting to consider
%formulations where a pair of nodes is pseudo-separated. For example, QoS in communication networks.
% Discuss WIN-T. 
However, the functionality a network provides may break down even when elements of a network
remain connected. For example, suppose $G$ is a communication network, with edge lengths representing
transmission time delay over that edge. 
For nodes $s, t$ to communicate, it is necessary 
that the total time-delay on the routing path by which they communicate remain below
some threshold $T$. If the shortest-path distance between $s, t$ exceeds $T$, communication
breaks down, despite the fact that $s$ and $t$ are topologically connected within the network.
Another example is the shipping of a perishable item through a transportation network. If the
item reaches its destination after it has perished, it is of no use to the recipient. 
Therefore, instead of considering network failure to
occur if elements of the network are topologically separated, we
propose a more general measure of network failure: 
network functionality is impaired after the \emph{$T$-separation} of elements in a network, where $T$ is
a real number.  Two nodes $s,t$ are $T$-separated if the weighted shortest-path distance 
exceeds $T$.

%c. Give example showing difference between cut and pcut, and discuss how naive solutions,
%like just considering T-component containing s and doing min. cut, are ineffective.
As we demonstrate in this work, the $T$-separation analogue (defined below)
to the classical CUT problem cannot be reduced to CUT unless $P = NP$.
Given a constant $T$, the \emph{minimum $T$-pseudocut} (T-PCUT) problem
takes as input a directed graph $G$, targeted pair $(s,t)$, 
and distance function $d$ on the edges of $G$. 
The problem asks
for the minimum-size set of vertices (edges) $W$ to remove from $G$, such
that after the removal of $W$, the $d$-shortest paths distance $d(s,t) > T$. 
To demonstrate the differences between CUT and T-PCUT, consider
the following example. Let $G$ be the network shown in Fig. \ref{fig:example},
let $s = 0, t = 12$, and consider $d(e) = 1$ for each edge $e \in G$; finally,
set $T = 5$. An optimal solution to the vertex version of CUT (also known
as minimum vertex separator \cite{Vazirani2001}) must contain three nodes,
while the removal of $W =\{ 5, 7 \}$ is an optimal solution to this instance of $5$-PCUT;
after removal of $W$, $d(s,t) = 6 > T$. Observe that the naive proposal of eliminating
all vertices of distance greater than $T$ from $s$ and then solving CUT on the new
graph does not work, since every node $v$ in $G$ initially satisfies $d(s, v) \le 4$.

%c1
Although the new combinatorial problems we propose in this work should
be broadly applicable, the application in which we are most interested 
is structural vulnerability with respect to 
additive Quality-of-Service (QoS) metrics on 
communication networks. For example, the total time-delay, jitter, or
packet-loss\footnote{Packet-loss can be converted to an additive metric,
as described in Lemma \ref{lem:per}.} between two nodes
in a communication network are additive QoS metrics. For a given additive
QoS metric $Q$, the minimum acceptable threshold $T_Q$ for this metric
is a constant independent of any particular communication network, 
although it will vary with the desired communication application, 
such as voice or video call, process control, or machine control.
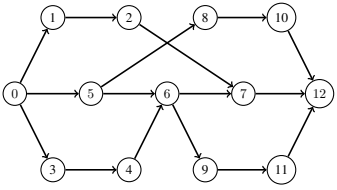
\begin{figure}
  \centering
%\subfigure[] {
  \resizebox{0.25\textwidth}{!}{
    \begin{tikzpicture}
      \begin{scope}[every node/.style={circle,thick,draw}]
    \node (0) at (0,0) {$0$};
    \node (5) at (2,0) {$5$};
    \node (6) at (4,0) {$6$};
    \node (7) at (6,0) {$7$};
    \node (12) at (8,0) {$12$};
    \node (1) at (1,2) {1};
    \node (2) at (3,2) {2};
    \node (8) at (5,2) {8};
    \node (10) at (7,2) {10};
    \node (3) at (1,-2) {3};
    \node (4) at (3,-2) {4};
    \node (9) at (5,-2) {9};
    \node (11) at (7,-2) {11};
\end{scope}

\begin{scope}[every node/.style={fill=white,circle},
              every edge/.style={draw=black,very thick}]
  \path [->] (0) edge (5);
  \path [->] (5) edge (6);
  \path [->] (6) edge (7);
  \path [->] (7) edge (12);
  \path [->] (0) edge (3);
  \path [->] (3) edge (4);
  \path [->] (4) edge (6);
  \path [->] (6) edge (9);
  \path [->] (9) edge (11);
  \path [->] (11) edge (12);
  \path [->] (5) edge (8);
  \path [->] (8) edge (10);
  \path [->] (10) edge (12);
  \path [->] (0) edge (1);
  \path [->] (1) edge (2);
  \path [->] (2) edge (7);
\end{scope}
\end{tikzpicture} } \label{fig:example}
%}
  % \subfigure[] {
  %   \resizebox{0.18\textwidth}{!}{
  %     \begin{tikzpicture}
  %       \begin{scope}[every node/.style={fill=white,circle}]
  %         \node (1) at (0,0) {$1$};
  %         \node (2) at (1,0) {$2$};
  %         \node (3) at (2,0) {$\ldots$};
  %         \node (n) at (3,0) {$n$};
  %         \node (s) at (1,2) {$s$};
  %         \node (t) at (2,-2) {$t$};
  %       \end{scope}

  %       \begin{scope}[every node/.style={fill=white,circle},
  %           every edge/.style={draw=black,very thick}]
  %         \path [->] (s) edge (1);
  %         \path [->] (s) edge (2);
  %         \path [->] (s) edge (n);
  %         \path [->] (1) edge (t);
  %         \path [->] (2) edge (t);
  %         \path [->] (n) edge (t);
  %       \end{scope}
  %   \end{tikzpicture} } \label{fig:gap}
  % }
\caption{A graph $G$ exemplifying the necessity of new solutions for the PCUT problem, as explained in the text. }
%(b):  The construction from Theorem \ref{thm:gap} showing integrality gap of $\Omega (n)$ for IP \ref{IP_beta}  } 
\end{figure}

%d. Our contributions
\subsection{Our contributions}
\begin{itemize}
  \item We introduce $T$-separation analogues to the following 
    two classical
    combinatorial problems: the CUT problem defined above,
    and the MULTI-CUT problem \cite{Leighton1999}, in which
    $k$ pairs $(s_1, t_1), \ldots, (s_k, t_k)$ must be
    disconnected with minimum number of edges (nodes)
    removed.
    %, and (3) the PARTIAL MULTI-CUT problem \cite{Golovin2009},
    %where, instead of separating all $k$ pairs, only some fraction
    %of them must be separated.
    Collectively, we refer to these new formulations
    as \emph{pseudocut} problems, and they are
    respectively T-PCUT and T-MULTI-PCUT;
    %, and PARTIAL T-MULTI-PCUT;
    these problems are formally defined in
    Section \ref{sect:prob}. 
  \item \emph{Computational complexity:} 
    We show that with arbitrary edge weights, $T$-PCUT is
    $NP$-complete. With uniform edge weights, we show $T$-MULTI-PCUT 
    is inapproximable within a factor of 1.3606 by approximation-preserving
    reduction from the minimum vertex cover problem.  
%    For PARTIAL T-MULTI-PCUT, the situation is
%    worse; we show that the natural integer programming formulation
%    of this problem has integrality gap $\Omega (n)$, which provides
%    evidence against the existence of non-trival approximation algorithms,
%    where $n$ is the number of nodes in the input graph $G$.
  \item \emph{Approximation algorithms:} For the T-PCUT and T-MULTI-PCUT
    problems with uniform edge weights, 
    we provide GEN, an $O( \log n )$-approximation algorithm; and
    FEN, a $(T + 1)$-approximation algorithm. In addition, we provide 
    GEST, an efficient, randomized algorithm with probabilistic performance guarantee:
    with probability $1 - 1 / n$, GEST returns a feasible solution with
    cost within ratio $O( \alpha \delta^T + \log k )$ of optimal, where
    $k$ is the number of pairs to $T$-separate, $\delta$ is the maximum
    degree in the graph, and $\alpha$ is user-defined parameter in $(0,1)$. 
    The time complexity of GEST is $O( k^3n\log (2n^2) / 2\alpha^2 )$, so $\alpha$
    gives the user control of the trade-off between performance and running time.
    %Finally, for the PARTIAL T-MULTI-PCUT
    %we provide bicriteria pseudo-approximation algorithms with performance ratio
    %of $(T + 1) / \epsilon$ and $O( \log n / \epsilon )$, where $\epsilon$ is 
    %a user-defined parameter in $(0,1)$.
  \item \emph{Vulnerability assessment:} 
    Finally, we utilize the pseudocut problems 
    to formulate a vulnerability assessment for an arbitrary additive
    QoS metric on communication networks. 
    We then perform
    extensive experimental evaluations of our algorithms
    in the framework of this vulnerability assessment.
\end{itemize}
\subsection{Related work}
%e. Related work (I think this section should be part of intro., may have to make case
%to Dr. Thai.
The theoretical results for min-cut, multi-cut, and partial multi-cut
vary depending on whether the edge or vertex version of the problem
is considered, and whether the graph is undirected or directed.
Table \ref{table:related} 
shows the current status of the best-known approximation ratios 
for each version of the problem, and the references where a proof of this ratio
may be found.
% in this table,
%MULTI and PARTIAL are abbreviated as M and P, respectively.
In contrast, our algorithms work equally well in undirected
or directed graphs and for the vertex or edge version of the pseudocut
problem. To the best of our knowledge, we are the first to consider
the pseudocut problems.

The seminal work of Ford and Fulkerson showed the max-flow and min-cut are
equal for the CUT problem \cite{Ford1956}. Leighton and Rao showed an
analogous result for the multi-cut problem \cite{Leighton1999} using
multicommodity max-flow, which gives $O( \log k )$-approximation algorithm
for the edge version of multi-cut problem in undirected graphs. For the
node version of multi-cut in undirected graphs, Garg et al.  \cite{Garg1994} gave
an $O( \log k )$-approximation algorithm.
For the edge version of multi-cut in directed graphs, 
Cheriyan et al. \cite{Cheriyan} gave an $O( \sqrt{n \log k})$-approximation;
Gupta \cite{Gupta2003} improved this ratio to $O( \sqrt{n} )$,
and finally Agarwal et al. \cite{Agarwal2007} improved the ratio to $O( n^{11/23})$.
For multi-cut in trees, Garg et al. \cite{Garg1993a} 
provided another max-flow min-cut
relationship, giving a $2$-approximation for multi-cut in trees.
%Levin et al. \cite{Levin2006} gave an $(8 / 3 + \epsilon)$-approximation for partial
%multi-cut on a tree. Separately, Golovin et al. \cite{Golovin2009} gave 
%$(8 /3 + \epsilon)$-approximation for the same problem, and showed
%how it could be used to give $O( \log^2 n \log \log n)$-approximation
%for partial multicut on general undirected graphs.

%With respect to network reliability, the $\beta$-vertex disruptor
%formulations of Thang et al. \cite{Dinh2012b,Dinh2015a} are recovered
%in our global vulnerability assessment if $T = \infty$.
A QoS-aware vulnerability assessment has been considered in Xuan et al.
\cite{Xuan2010}; however, the complexity of their assessment lies above 
even the $NP$ class as a valid solution cannot even be checked in polynomial time. 
A related problem to the single pair T-PCUT was studied by Israeli and Wood \cite{Israeli2002};
in this problem (MSP), given a fixed budget $k$ and pair $(s,t)$, a set of $k$ edges are
sought to maximize the shortest path between $s,t$. Israeli and Wood seek exact 
solutions using a bilevel optimization model, and this problem has been used 
as the basis for the detection of critical infrastructure and network vulnerability \cite{Scaparra2008,Brown2006}.
However, we emphasize the difference between T-PCUT and MSP: in T-PCUT, it is the
size (or cost) of the critical set that must be minimized; furthermore, MSP is 
formulated for edge interdiction only, while we primarily consider node interdiction.
%Furthermore, the pseudocut problems are more appropriate for our
%motivation of structural vulnerability with respect to a QoS requirement. 
Finally,
we have found only expensive exact methods to solve MSP; to the best of our knowledge,
no efficient solutions MSP with performance guarantee exist.
\begin{table}[htb]
\centering
\caption{Approximation results} \label{table:related}
%\begin{adjustwidth}{-0.5cm}{0cm}
\begin{tabular}{*3c}
\hline
Problem  &  Undirected & Directed \\
\hline
%CUT   & $\longrightarrow$ &  \multicolumn{2}{c}{1} & $\longleftarrow$ \\
CUT (both)  & 1 & 1 \\
M-CUT (edge) &  $O( \log k )$ \cite{Vazirani2001}  & $O(n^{11/23})$ \cite{Agarwal2007} \\
%P. M-CUT (edge) & $O( \log^2 n \log \log n )$ \cite{Konemann2011} & - \\
M-CUT (vertex)  & $O( \log k )$ \cite{Garg1994} & - \\ 
%P. M-CUT (vertex) & -  & - \\ 
\hline \hline
T-PCUT (both)  & $T + 1$ & $T + 1$ \\
T-M-PCUT (both) & $T + 1$ & $T + 1$ \\
%P. T-M-PCUT (both) & $(T + 1) / \epsilon $ & $(T + 1)/ \epsilon $ 
\end{tabular}
%\end{adjustwidth}
\end{table}

%% \emph{Related work: } 
%% For all of these, need to
%% consider the status of the vertex version of the problem.
%% \begin{itemize}
%% \item CUT. Solvable in polynomial time.
%% \item MULTI-CUT. 
%%   Garg et al. \cite{Garg1994} give $O( \log k )$-approximation ($k = | \mathcal{S} |$)
%%   for node version of multi-cut for weighted nodes on \emph{undirected} graphs. 

%% \item PARTIAL MULTI-CUT.
%% \end{itemize}

\subsection{Organization}
The rest of this paper is organized as follows.
In Section \ref{sect:prob}, we define the pseudocut problems,
discuss motivating applications, define the QoS
vulnerability assessment, and formulate the pseudocut problems
as integer programs. In Section \ref{sect:complexity},
we analyze the computational complexity of the pseudocut problems.
In Section \ref{sect:approx_algs}, we present our three approximation
algorithms. In Section \ref{sect:exp},
we experimentally evaluate our algorithms in the context of the QoS vulnerability
assessments. Finally, in Section \ref{sect:conclusion}, we summarize
our contributions and discuss future work.

\section{Problem definitions} \label{sect:prob}
In this section, we introduce the vertex versions of
the pseudocut problems; the edge versions are presented
in Appendix \ref{apx:edge}.
Let $T$ be an arbitrary but fixed constant throughout this section.
The problems will take as input a triple $(G,c,d)$,
where $G$ is a directed graph $G = (V,E)$;
$c: V \to \mathbf{R}^+$ is a cost function on vertices
representing the difficulty of removing
each node; and $d: E \to \mathbf{R}^+$ is a length function on
edges. For example, $d(e)$ could be the latency or packet loss
on edge $e$. Although both $c$ and $d$ may be considered weight
functions, we use \emph{cost} for $c$ and \emph{length} for $d$
to avoid confusion.
The case when $c(v) = 1$ for all vertices is referred to
as \emph{uniform cost},
and the case when $d(e) = 1$ for all edges
is referred to as \emph{uniform length}.
The distance $d(u, v)$ between two vertices is the
length of the $d$-weighted, 
directed, and shortest path between $u$ and $v$; the
cost $c(W)$ of set $W$ of a set of vertices is the
sum of the costs of individual vertices in $W$.
\begin{prob}[Minimum $T$-pseudocut (T-PCUT)]
  Given triple $(G,c,d)$ and a pair $(s,t)$ of vertices of $G$, 
  determine a minimum cost set $W \subset V \backslash \{s, t\}$ of
  vertices such that
  $d( s, t ) > T$ after the removal of $W$ from $G$.
  %$(V \backslash W, E)$.
  \label{prob:min-pcut}
\end{prob}

Notice that in the formulation of T-PCUT, we disallow the pair endpoints
to be chosen in the solution -- for the non-uniform cost version,
this restriction is unnecessary since the
endpoints could be assigned higher cost; however, we include this
restriction since otherwise the optimal solution would be trivial
in the uniform cost version.

\begin{prob}[Minimum $T$-multi-pseudocut (T-MULTI-PCUT)]
  Given triple $(G,c,d)$, and a target set of pairs of vertices of $G$,
  $\mathcal{S} = \{ (s_1, t_1), (s_2, t_2), \ldots, (s_k, t_k) \}$,
  determine a minimum cost set $W$ of vertices such that
  $d( s_i, t_i ) > T$ for all $i$ after the
  removal of $W$ from $G$. \label{prob:multi-pcut}
\end{prob}

% \begin{prob}[Minimum partial $T$-multi-pseudocut (PARTIAL T-MULTI-PCUT)]
%   Given triple $(G,c,d)$, a target set of pairs of vertices of $G$,
%   $\mathcal{S} = \{ (s_1, t_1), (s_2, t_2), \ldots, (s_k, t_k) \}$, and
%   $\beta \in [0,1]$, 
%   determine a minimum cost set $W$ of vertices such that
%   $d( s_i, t_i ) > T$ for at least $\beta k$ pairs $(s_i, t_i)$ after the
%   removal of $W$ from $G$. \label{prob:partial-multi-pcut}
% \end{prob}

In contrast to T-PCUT, we allow 
picking members of
pairs in $\mathcal{S}$ into the solution of T-MULTI-PCUT; 
thus, there is always a feasible solution of size at most $k$. If a vertex $v$ is removed from $G$, we adopt the convention that $d(v,w) = \infty$ for all vertices $w \in G$.

In the above two formulations, we emphasize again that
the threshold $T$ is a fixed
constant independent of the input; in addition,
we introduce versions of these problems where $T$ is part
of the input. 
We will refer to the versions of these
problems where $T$ is an input as PCUT and MULTI-PCUT, respectively.
Finally, the algorithms in Section
\ref{sect:approx_algs} generalize to the 
edge versions of the problems as well, as discussed in
Appendix \ref{apx:ev-algs}.

\subsection{Motivation and applications for the pseudocut problems} \label{sect:motivation}
In this section, we give brief overviews of two potential applications
of the pseudocut problems. Motivated by these examples, we next provide the
vulnerability assessment for QoS on communication networks.

\subsubsection{Industrial Internet of Things} \label{sect:IIoT}
An emerging application for pseudocut problems is
the Industrial Internet of Things (IIoT). As everyday objects
become increasingly equipped with means for electronic
identification and communication, 
from Radio Frequency Identification (RFID) to smarter communication
capabilities, new applications and scenarios have emerged
in the Internet of Things \cite{Atzori2010,Xu2014}.

As surveyed in \cite{Sadeghi2015},
an emerging trend is to integrate communication capabilities into
industrial production systems. Such 
cyberphysical systems (CPS) in the production process 
are connected to conventional business IT networks. 
Integrated CPS 
allow extensive monitoring and control of production
facilities in real time. However, the QoS requirements
for control of production systems are very strict, and 
special routing protocols have been formulated to guarantee acceptable
QoS conditions \cite{Thrybom2009}. An IEEE task group
on Time-Sensitive Networking (TSN)
\cite{InstituteofElectricalandElectronicsEngineers} is
currently chartered to provide specifications to
allow time-synchronized low latency streaming services
through 802 networks. Critical data streams are guaranteed
certain end-to-end QoS by resource reservation; this
service is intended for industrial applications such
as process control, machine control, and vehicles;
and for audio/video streams.

As an example application for the T-PCUT, consider 
two nodes in IIoT as described above: $s$, a control node,
and $t$ a lower-level node. Further, suppose that
an acceptable level of packet loss ratio between $s,t$ is 
$10^{-10}$. Then, the problem instance of T-PCUT is
the IIoT network $G$, with edges $e$ weighted by the 
metric $d$ defined in Lemma \ref{lem:per} below. A solution to $10^{-10}$-PCUT problem 
for $(s,t)$ identifies
the most critical vertices whose proper functioning is required 
to ensure $p(s,t) < 10^{-10}$, where $p(s,t)$
is the cumulative packet
loss ratio between $s$ and $t$. 

To convert the packet error rate between
nodes to an additive metric, we define
the following transformation.
Given network $G = (V,E)$, let $p_{uv} \in [0,1]$ represent 
packet error rate for each edge $(u,v) \in E$.
Then, the transformation is 
\begin{equation} \label{eq:per}
  p_{uv} \rightarrow - \log \left(1 - p_{uv} \right).
\end{equation}
%The following lemma is proved in Appendix \ref{apx:per}.
\begin{lemma} \label{lem:per}
  Let $p_{uv}$ represent packet error rate between each $(u, v) \in E$. 
  Then the transformation (\ref{eq:per}) yields an additive metric
  $d$ such that $1 - \exp \left(- d(s,t) \right)$ is the lowest cumulative
  packet error rate between nodes $s,t$ over all possible routing paths.
\end{lemma}
\begin{proof} %[Proof of Lemma \ref{lem:per}]
Let $G=(V,E)$ with packet error rate
$p_{er}(e) \in (0,1)$ be given for each $e \in E$.
Let $d(e) = - \log ( 1 - p_{er}(e) )$. Let $s,t \in G$,
and $\mathscr{P}$ be the set of all paths in $G$
from $s$ to $t$.
Then
\begin{align*}
  d(s,t) &= \min_{p \in \mathscr{P}} \sum_{e \in p} d(e) \\
  &= \min_{p \in \mathscr{P}} \sum_{e \in p} - \log (1 - p_{er}(e)) \\
  &= - \max_{p \in \mathscr{P}} \log \prod_{e \in p} (1 - p_{er}(e)) 
%  &= - \max_{p \in \mathscr{P}} \prod_{e \in p} (1 - p_{er}(e))
\end{align*}
Now, $\prod_{e \in p} (1 - p_{er}(e))$ is the probability a packet is successfully
transmitted along path $p$. Thus, maximizing this probability over all paths
minimizes both $d(s,t)$ and the cumulative packet error rate between $s,t$.

Furthermore, if packet error rate threshold $P$ is given, then by similar
reasoning
\[ d(s,t) < - \log (1 - P) \iff p_{er}(s,t) < P, \] 
where $p_{er}(s,t)$ is the cumulative packet error rate
between $s,t$.
\end{proof}
\subsubsection{Military communications networks}
Next generation millitary communications networks will be 
multilayer, interdependent networks \cite{Xu2001,Ali2013,Juarez2006}
comprising wired fiber-optic and wireless components, including
satellite communications.
For example, consider
the proposed Army Warfighter Information Network-Tactical (WIN-T) network, 
the theory of operation for which is contained in \cite{Ali2013}. 
WIN-T comprises interdependent wireless and wired
components that are organized into layers;
the WIN-T multi-tiered architecture is organized as follows:
(1) the space layer, utilizing military
satellite communications (MILSATCOM) and commercial
satellite bands, (2) the airborne layer, consisting of unmanned
aerial vehicles (UAVs), (3) the ground layer, which contains many different
kinds of nodes. These nodes communicate to each other 
and nodes in the other layers in
a variety of ways including wired LANs, wireless WANs,
and satellite communications.

To ensure QoS in WIN-T, traffic is only admitted
to the WAN network when the network infrastructure and
congestion state offer a high probability that the
traffic can be delivered within QoS requirements specified
in WIN-T Baseline Requirements Document. Thus, communication
failure
between a pair $s,t$ of nodes in the network may occur
despite the existence of a routing path between $s$ and $t$
in the network, if any of the QoS metrics are greater than
a threshold $T$.

Therefore, the T-PCUT problem would identify the most critical
nodes if communication between a given pair of nodes $(s,t)$. For
example, $s$ could be a commanding node attempting to send an
order to infantry unit $t$. If communication between $s$ and $t$
is a high priority, critical nodes identified by T-PCUT would be
especially important to protect against an adversarial attack.

\subsubsection{Vulnerability assessment on communication networks} \label{sect:qos_vuln}
Motivated by the above two examples, 
we present a vulnerability assessment for communication
networks in this section.
Let $C = (V,E)$ represent a communication network. We fix an additive
QoS metric $Q$ on the edges of $C$. 
Since the QoS metric $Q$ is additive,
we define the QoS metric on the path 
$p = p_0p_1 \cdots p_l \in C$ as
\[ Q(p) = \sum_{i = 1}^l Q( p_{i - 1}, p_i ). \]
Furthermore,
we denote the metric between a pair $s,t$ as
$Q(s,t)$, the shortest-path distance between
$s,t$, where the weight of each edge in the network
is $Q(u,v)$. Clearly, no routing path could provide better
QoS with respect to $Q$ than the $Q$-shortest path. 
Let $T$ be a constant representing the threshold such
that if $Q(s,t) > T$ then communication between $s$ and $t$ is 
no longer possible.
Notice that since the value of $Q$ on each edge 
is determined by network parameters,
it has a minimum value $q_{min}$ which is a constant independent
of the network size.

Next, we define the problems of identification
of the most critical elements of the network 
with respect to the metric $Q$ and threshold $T$, and
a given targeted set of pairs $\mathcal{S}$ in the network, with
respect to $T$-separation.
\begin{prob}[Targeted Communication Vulnerability Assessment (TCVA)]
  Given communication network $C = (V,E)$, an additive quality of service metric
  $Q$, a threshold $T$ for $Q$ indicating the highest acceptable value
  of $Q$ for communication between a pair of nodes in $C$, a
  targeted set $\mathcal{S} = \{(s_1,t_1),(s_2,t_2), \ldots, (s_k,t_k)\}$, and a cost function $c$ on $C$,
  determine $W \subset V$ of minimum cost
  such if $W$ is removed from $C$, then for all $(u,v) \in \mathcal{S}$,
   $Q(u,v) > T.$
\end{prob}
Notice that TCVA is exactly the T-MULTI-PCUT problem with the edge 
length function equal to the QoS value on the edge. 
% Next,
% we introduce a global vulnerability assessment, which seeks
% to identify the most important vertices with respect to $T$-separation
% of a $\beta$-fraction of all pairs in the network, where $\beta \in (0,1)$
% is a user-defined parameter. The larger the value of $\beta$, the 
% more pairs in the network are $T$-separated.
% \begin{prob}[Global Communication Vulnerability Assessment (GCVA)]
%   Given communication network $C = (V,E)$, 
%   $\beta \in [0,1]$, and cost function $c$ on $V$,
%   determine $W \subset V$ of minimum cost
%   such that at least a $\beta$-fraction of
%   pairs $(u,v)$ in $V$
%   satisfy
%   $Q(u,v) > T$, 
%   after the removal of $W$.
% \end{prob}
% Notice that GCVA is a special case of PARTIAL T-MULTI-PCUT with
% $\mathcal{S} = \{ \text{all pairs of vertices in $C$} \}$.

\subsection{Integer programming formulations}
In this section, we formulate the pseudocut problems as integer programs.
We will state the formulations for the pseudocut versions 
where $T$ is an input, but the same formulations apply when $T$ is a constant.
%\subsubsection{ PCUT and MULTI-PCUT }
We formulate PCUT and MULTI-PCUT as 
integer programs in the following way.
Let an instance $(G,c,d,\mathcal{S}, T)$ of MULTI-PCUT be given.
We will consider simple paths $p = p_0 p_1 \ldots p_l \in G$; that
is, paths containing no cycles. 
Let $\mathcal{P}( s_i , t_i )$ denote the set of simple paths $p$
between $(s_i,t_i) \in \mathcal{S}$ that satisfy the 
condition $d(p) \le T$. If a vertex $u$ lies on path $p$,
we write $u \in p$. The following lemma relates the
optimal solution to MULTI-PCUT to the minimum-size
hitting set of $\mathcal{P} = \bigcup_{i = 1}^k \mathcal{P}(s_i,t_i),$
which is necessary for the integer programming formulation.
%To formulate the integer programs, we 
%require the following lemma.
\begin{lemma} \label{lemm:cov_eq}
  Let $W^*$ be an optimal solution to an instance of
  MULTI-PCUT. Let $W'$ be a minimum cost 
  set of vertices satisfying $W' \cap p \neq \emptyset$
  for all $p \in \mathcal{P}( s_i, t_i )$ for all
  $(s_i, t_i) \in \mathcal{S}$.
  Then, $c( W' ) = c( W^* )$.
\end{lemma}
\begin{proof}
  Since $W^*$ is a solution to the MULTI-PCUT
  problem, we have $d( u, v ) > T$ for all 
  $(u, v) \in \mathcal{S}$ after the removal of
  $W^*$. Any path $p$ in $G$ between a pair
  $(u,v) \in \mathcal{S}$ satisfying $d(p) \le T$ must
  therefore satisfy $p \cap W^* \neq \emptyset$, for
  otherwise $d(u,v) \le T$. Thus, $c(W') \le c(W^*)$.

  Similarly, the removal of $W'$ from $G$
  ensures $d(u,v) > T$ for all $(u,v) \in \mathcal{S}$,
  hence $c(W^*) \le c(W')$.
\end{proof}

As a consequence of Lemma \ref{lemm:cov_eq}, we can
formulate MULTI-PCUT as a covering integer program.
Consider the vertex set of $G$ to be $\{1, \ldots, n \}$.
Let $A^{(u,v)}_{p, i} = 1$ if vertex $i$ lies on path
$p \in \mathcal{P}(u, v)$, where $(u,v) \in \mathcal{S}$. 
If $i \not \in p$, let $A^{(u,v)}_{p, i} = 0$.
Also, let variable $w_i = 1$ if vertex $i$ is to be
chosen into the set of vertices $W$, 
and $0$ otherwise. Finally, denote the cost of choosing
vertex $i$ as $c_i$, and let vectors $w = (w_1, \ldots, w_n)$ and
$c = (c_1, \ldots, c_n)$.
Then,
the covering $0-1$ integer program formulation 
is as follows.

\begin{IP}[IP \ref{IP_mp}] \label{IP_mp}
\begin{align} 
  & \min c \cdot w \text{ such that } \nonumber \\
  &\sum_{i=1}^n A_{p,i}^{(u,v)} w_i \ge 1, \, \forall p \in P(u,v), \, \forall (u,v) \in \mathcal{S} \label{IP_cov}\\
  & w_i \in \{ 0, 1 \}, \, \forall i \in \{1, \ldots, n \} \label{IP_int}
\end{align}
\end{IP}

The constraints (\ref{IP_cov}) ensure that for each
path $p \in \mathcal{P}(u,v)$, we choose at least one node $i \in p$.
By Lemma \ref{lemm:cov_eq}, the optimal solution to IP \ref{IP_mp}
corresponds to an optimal solution of MULTI-PCUT. 
%For smaller
%network sizes, it may be feasible to solve IP \ref{IP_mp} exactly.
The linear relaxation of IP \ref{IP_mp} is designated
LP \ref{IP_mp}, in which each constraint (\ref{IP_int}) is
replaced by $w_i \in [0,1]$.
Finally, we remark 
that since PCUT is a special case of MULTI-PCUT,
IP \ref{IP_mp} and all solutions we discuss apply to PCUT as well.

\subsubsection{Discussion}
Notice that if we let $T$ become large enough, the classical
problems CUT and MULTI-CUT are recovered
from PCUT and MULTI-PCUT.%, and PARTIAL MULTI-PCUT.

If $T$ is an input, IP \ref{IP_mp}
above is superpolynomial 
in size; there could be $n^T$ constraints (1);
The analogous integer program for MULTI-CUT
also could have exponentially  many constraints
but has a polynomial-time separation oracle that enables the 
linear relaxation to be solved in polynomial time by the
ellipsoid method \cite{Vazirani2001}. However, this separation oracle
does not work for the linear relaxation of IP 1;
in general,
the linear relaxation may not be solvable in polynomial time.
However, the IP formulations above hold when $T$ is
a constant. Thus, IP \ref{IP_mp} is polynomial in size
when T-MULTI-PCUT is considered.

Finally, notice that not all instances to PCUT admit a valid
solution; suppose as input a graph
consisting of a single edge $(s,t)$ is given.
PCUT is formulated to disallow choosing $s$ or $t$; hence, there is no
solution. Whether a feasible solution exists can easily be detected
in polynomial time, so unless otherwise stated, 
we assume that a feasible problem instance is given in our analysis.

\section{Computational complexity} \label{sect:complexity}
In this section, we present our results on the computational complexity of
the pseudocut problems. 
%\subsection{Uniform edge length} \label{sect:uniform}

\subsection{T-PCUT} \label{sect:uniform_pcut} 
We give polynomial-time algorithms
for certain cases of the version of T-PCUT with 
uniform lengths. However, T-PCUT
with arbitrary edge lengths and uniform vertex costs 
is shown to be $NP$-hard. 
%The proofs of Propositions \ref{prop:T3}, \ref{prop:delta},
%and Theorem \ref{thm:NPc} are given in
%Appendix \ref{apx:proofs}.
\begin{proposition} \label{prop:T3}
  For $T \le 3$, T-PCUT with uniform lengths and costs
  is solvable in polynomial time.
\end{proposition}
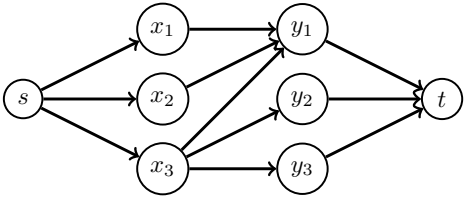
\begin{figure}
  \centering
  \resizebox{0.35\textwidth}{!}{
  \begin{tikzpicture}
\begin{scope}[every node/.style={circle,thick,draw}]
    \node (s) at (0,0) {$s$};
    \node (x1) at (2,1) {$x_1$};
    \node (x2) at (2,0) {$x_2$};
    \node (x3) at (2,-1) {$x_3$};
    \node (y1) at (4,1) {$y_1$};
    \node (y2) at (4,0) {$y_2$};
    \node (y3) at (4,-1) {$y_3$};
    \node (t) at (6,0) {$t$};
\end{scope}

\begin{scope}[every node/.style={fill=white,circle},
              every edge/.style={draw=black,very thick}]
  \path [->] (s) edge (x1);
  \path [->] (s) edge (x2);
  \path [->] (s) edge (x3);
  \path [->] (x1) edge (y1);
  \path [->] (x2) edge (y1);
  \path [->] (x3) edge (y1);
  \path [->] (x3) edge (y2);
  \path [->] (x3) edge (y3);
  \path [->] (y1) edge (t);
  \path [->] (y2) edge (t);
  \path [->] (y3) edge (t);
\end{scope}
\end{tikzpicture} }
\caption{An example of $G'$ in the
  analysis of Proposition \ref{prop:T3}.}

\label{fig:T=3}
\end{figure}
\begin{proof} %[Proof of Proposition \ref{prop:T3}]
  Let $G, (s,t)$ be an instance of T-PCUT.
  First consider the case $T = 2$. Since
  edge lengths are uniform, all paths $p$ of length $2$ from $s$
  to $t$ have exactly three vertices: $p = sxt$ for some $x \in V$.
  Therefore, no such paths can intersect unless they are identically
  equal. So to ensure $d(s,t) > 2$, one must simply remove all
  intermediate vertices between $s$ and $t$. 

  Next, suppose $T = 3$. Let $p_1 = sxyt$ be a path of length 3 
  from $s$ to $t$, and let $p_2$ be a path of length 2 that intersects
  $p_1$. In order to satisy $d(s,t) > 3$, $p_2$ must be broken, which
  can happen in only one way and
  necessarily breaks $p_1$ as well. Hence, in the first step we
  break all paths of length 2 in the same way as for the $T = 2$ case,
  and denote the modified graph as $G'$.
  The remaining paths of length 3 do not intersect paths of length 2. 
  Two distinct paths of length 3 can intersect each other in a maximum
  of one vertex. Let $X = \{x_1, x_2, x_3, \ldots \}$ be the set of all
  nodes that appear as the second node (after $s$) on a path of length 3; similarly,
  let $Y = \{y_1, y_2, y_3, \ldots \}$ be the set of nodes appearing as the third node
  on a path of length 3. Notice that $X \cap Y = \emptyset$, because otherwise
  a path of length 2 would still be extant in the graph, but all such paths were
  removed in the first step.

  Thus, the relevant subgraph $G'$ will appear of the form
  exemplified in
  Fig. \ref{fig:T=3}. Notice that an edge $(x_2, x_1)$ would have no relevance
  to the solution, as the only way to create a path of length 3 using
  $(x_2,x_1)$
  would be to add $(x_1,t)$ as well; but this process creates the path
  $sx_1t$, which is of length 2; so $x_1$ would have been chosen in the first
  step. If we delete $s$ and $t$ from the graph $G'$, we see that our
  problem reduces to a bipartite vertex cover problem, which is
  solvable in polynomial time; the second step will consist of the optimal
  solution to this problem.  The final solution is the union of vertices
  chosen in the first and second steps.
\end{proof}

\begin{proposition} \label{prop:delta}
  Let $D$ be a constant, 
  T-PCUT$( G, (s,t) )$ be an instance of T-PCUT for some
  constant $T$ with uniform lengths and uniform costs. 
  If the maximum degree
  $\delta$ in $G$ satisfies $\delta \le D$, then
  the optimal solution $W$ is computable in polynomial
  time.
\end{proposition}
\begin{proof} %[Proof of Proposition \ref{prop:delta}] 
Consider all distinct paths of length at most $T$ 
starting from $s$ and ending at $t$.  
The number of distinct vertices on these paths is 
$O( \delta^T ) = O( D^T )$; let us call this set $V'$. 
Therefore, the number of possible subsets of these vertices is a constant
bounded by $O( 2^{D^T} )$. Since each subset
can be checked in polynomial time, the optimal solution can
be found by checking each possible subset of $V'$.
\end{proof}

\begin{figure}
  \centering
  \begin{tikzpicture}
\begin{scope}[every node/.style={circle,thick,draw}]
    \node (A) at (0,0) {$u_1$};
    \node (A1) at (0,-2) {$v_1$};
    \node (B) at (2.5,0) {$u_2$};
    \node (B1) at (2.5,-2) {$v_2$};
    \node (C) at (5,0) {$u_n$};
    \node (C1) at (5,-2) {$v_n$};
    \node (D) at (7.5,0) {$u_{n+1}$} ;
\end{scope}

\begin{scope}[every node/.style={fill=white,circle},
              every edge/.style={draw=black,very thick}]
    \path [->] (A) edge node {$e_1$} (B);
    \path [->] (A) edge node {$f_1$} (A1);
    \path [->] (A1) edge node {$g_1$} (B);
    \path [->] (B) edge node {$\ldots$} (C);
    \path [->] (B) edge node {$f_2$} (B1);
    \path [->] (B1) edge node {$\ldots$} (C);
    \path [->] (C) edge node {$e_n$} (D);
    \path [->] (C) edge node {$f_n$} (C1);
    \path [->] (C1) edge node {$g_n$} (D);

\end{scope}
\end{tikzpicture}
  \caption{The construction for the reduction from Knapsack
    to Pseudocut} \label{fig:NP_hard}
\end{figure}
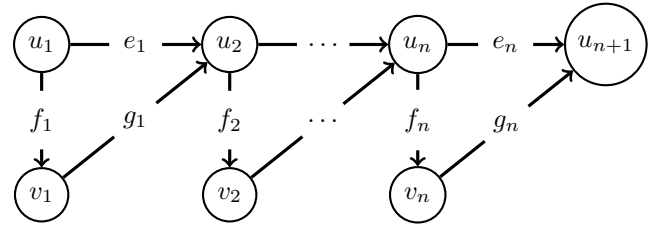
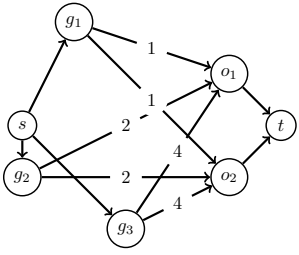
\begin{figure}
  \centering
%  \subfigure[] {
    \resizebox{0.22\textwidth}{!}{%
      \begin{tikzpicture}
        \begin{scope}[every node/.style={circle,thick,draw}]
          \node (s) at (0,0) {$s$};
          \node (g1) at (1,2) {$g_1$};
          \node (g2) at (0,-1) {$g_2$};
          \node (g3) at (2,-2) {$g_3$};
          \node (o1) at (4, 1) {$o_1$};
          \node (o2) at (4, -1) {$o_2$};
          \node (t) at (5,0) {$t$};
        \end{scope}

        \begin{scope}[every node/.style={fill=white,circle},
            every edge/.style={draw=black,very thick}]
          \path [->] (s) edge (g1);
          \path [->] (s) edge (g2);
          \path [->] (s) edge (g3);
          \path [->] (o1) edge (t);
          \path [->] (o2) edge (t);
          \path [->] (g3) edge node {$4$} (o1);
          \path [->] (g3) edge node {$4$} (o2);
          \path [->] (g2) edge node {$2$} (o1);
          \path [->] (g2) edge node {$2$} (o2);
          \path [->] (g1) edge node {$1$} (o1);
          \path [->] (g1) edge node {$1$} (o2);
        \end{scope}
    \end{tikzpicture} }

%  }

  \caption{ 
    Construction for $k=3$ 
    showing tightness of the ratio of Alg. \ref{alg:multipcut}.
    Numbers on certain edges indicate the number of
    disjoint paths of length $2$ between the corresponding nodes.
    Thus, there exactly $8$ paths $p$ from $s$ to $t$ passing through
    node $g_3$ satisfying $d(p) < 5$.}     \label{fig:tightness}
\end{figure}
\begin{theorem} \label{thm:NPc}
  Consider the decision version of 1-PCUT with uniform costs
  and arbitrary lengths; that is, given problem instance
  1-PCUT$(G, (s,t))$ with uniform costs and arbitrary lengths,
  and given constant $D > 0$, determine if a solution
  $W \subset V$ exists with $|W| \le D$. This problem
  is NP-complete.
\end{theorem}
\begin{proof} %[Proof of Theorem \ref{thm:NPc}]
For clarity, we first 
prove the theorem for the edge version of
1-PCUT (where edges $e \in G$ have both
cost and length functions), with arbitrary costs of edges;
next, we discuss how to modify the proof for the uniform
cost function and the vertex version of PCUT.
The decision problem is
clearly in $NP$.
To show $NP$-hardness, we first reduce
the Knapsack problem to an instance of
Pseudocut with non-uniform costs; then
we discuss how to modify the reduction
for uniform costs. A problem
instance of Knapsack is specified as follows. Let
$S = \{ a_1, \ldots, a_n \}$ be a set of objects
with sizes $w(a_i) \in \mathbf{Z}^+$ and profits
$p( a_i ) \in \mathbf{Z}^+$, and a ``knapsack capacity''
$W$, and desired profit $P$. The decision version of 
the problem is to find a subset of objects with
total profit at least $P$ and total size bounded
by $W$.

Given a Knapsack instance, we construct an instance of
the pseudocut problem in the following way. For each item
$a_i$, we add nodes $u_i, v_i$ and edges $e_i = (u_i, u_{i + 1})$,
$f_i = (u_i, v_i)$, and $g_i = (v_i, u_{i + 1})$. We also
set the following cost and $d$ values:
$c(e_i) := w( a_i )$, $d( e_i ) := 0$, $c( f_i ) := \infty$,
$d(f_i) := 0$, $c( g_i ) := \infty$, and $d( g_i ) := p( a_i ) / P$.
Fig \ref{fig:NP_hard}
illustrates this construction.

Then, letting $s = u_1$, $t = u_{n + 1}$, we have an instance
of the 1-PCUT, the decision version of which is 
whether there exists a set of edges of total cost at most $W$
such that $d(s, t) \ge 1$. Notice that including edge $e_i$
into a solution incurs cost $c ( a_i )$ and adds $p ( a_i ) / P$
to $d(s,t)$. Furthermore, edges $f_i$ and $g_i$ will not be chosen since
these edges have infinite cost. So choosing edge $e_i$
exactly corresponds to adding item $a_i$ into the knapsack, and
solutions to the Knapsack instance and the Pseudocut
instance are in one-to-one correspondence, with corresponding
solutions having the same cost. Also, $d(s,t) \ge 1$ iff
the corresponding solution to the Knapsack problem 
has profit at least $P$.

\emph{Modification for vertex version:}
To obtain the $NP$-hardness of the uniform cost vertex 1-PCUT
problem, we discuss how to modify the above reduction.
The first modification is to replace each vertex in the construction
with a clique of $W + 1$ vertices. Edges $f_i$ and $g_i$ are replaced by
$W + 1$ edges matching clique $v_1$ with $u_1$ and with $u_2$, respectively.
Instead of a single edge $e_i$ we add $c(a_i)$ vertices $w_{ij}$ between
$u_i$ and $u_{i+1}$, connecting each vertex in cliques $u_i, u_{i+1}$ to each $w_{ij}$.
Distinct nodes $s, t$ are added and $s$ is connected to each vertex in first clique $u_1$,
and $t$ to each node in clique $u_n$.
Thus, in order to add $p(a_i) / P$ to the distance $d(s,t)$, it is necessary to pick all
$c(a_i)$ vertices $w_{ij}$.
\end{proof}

\subsection{T-MULTI-PCUT} \label{sect:uniform_multi}
In this section, we show uniform length and cost T-MULTI-PCUT to be inapproximable within
a factor of $1.3606$.

\begin{theorem} \label{thm:nphard-multipcut} Let
  $T \ge 1$. Consider
  the decision version of T-MULTI-PCUT with uniform lengths
  and costs; that is, given
  problem instance
  T-MULTI-PCUT$(G, \mathcal{S})$ with uniform lengths and costs,
  determine if a solution $W \subset V$ exists with
  $|W| \le D$. This problem is NP-complete.
\end{theorem}
\begin{proof}
  The feasibility of a solution $W$ satisfying $|W| \le D$ can
  easily be checked in polynomial time, so T-MULTI-PCUT $\in NP$.
  We give an approximation-preserving reduction \cite{Vazirani2001}
  from the vertex cover problem to T-MULTI-PCUT.
  Let $H$ be an instance of the vertex cover problem;
  let the vertex set of $H$ be $V = \{ 1, 2, \ldots, n \}$. An instance
  of T-MULTI-PCUT s constructed as follows.
  Let $G$ be a complete graph on $\{1, 2, \ldots, n \}$,
  and $\mathcal{S}$ be the edge set of $H$.

  Then, there is a natural one-to-one, cost-preserving
  correspondence
  between solutions of the two instances; namely the identity
  mapping: if $W \subset V$ is a vertex cover of size $l$, 
  $W$ is also a feasible solution to the T-MULTI-PCUT instance
  of size $l$,
  since $(u,v) \in \mathcal{S}$ implies $(u,v) \in H$, which
  implies $u \in W$ or $v \in W$ since $W$ is a vertex cover,
  which finally implies $d(u,v) = \infty > T$ in $G$ (by the
  convention discussed in Section \ref{sect:prob}). 
  If $W \subset V$ is a solution to T-MULTI-PCUT, then
  for each $(u,v) \in H$, $d(u,v) = \infty$ after removal
  of $W$. Since the edge $(u,v)$ is in $G$, $u$ or $v$ is 
  in $W$, so that $W$ is a vertex cover.
\end{proof}
\begin{corollary} Unless $P = NP$, there is no 
  polynomial-time approximation to uniform length, cost
  T-MULTI-PCUT within a factor of $1.3606$, for $T \ge 1$.
\end{corollary} 
\begin{proof}
  This corollary follows from
  the proof of Theorem \ref{thm:nphard-multipcut} and
  the inapproximability of vertex cover \cite{Dinur2005}.
\end{proof}

\section{Approximation algorithms} \label{sect:approx_algs}
In this section, we present three approximation algorithms for
arbitrary vertex cost T-MULTI-PCUT, when the length function on edges
is bounded below: $d(e) > q_{min}$ for some constant $q_{min} > 0$.
In this case, we call the edge lengths \emph{bounded}.
Recall from Section \ref{sect:motivation} that edge
lengths are bounded
when the edge length function $d$ is an additive QoS
metric.
For the case of bounded edge lengths, we let constant $T_0 = T / q_{min}$. 
If the length function is uniform, then of course $T_0 = T$.
For bounded edge length, arbitrary vertex cost T-MULTI-PCUT, 
we present GEN, an $O( \log n )$-approximation algorithm, and FEN, a
$(T_0 + 1)$-approximation algorithm in Section \ref{sect:covering_approx}.
Although these algorithms run in polynomial time since $T_0$ is constant,
their running time may suffer if $T_0$ is large for some application.
Hence, we also present a randomized algorithm with probabilistic
performance guarantee in Section \ref{sect:gest}, capable of running
efficiently even for large $T_0$.

%Finally, due to the large integrality gap of the IP formulation
%of PARTIAL T-MULTI-PCUT, we present a bicriteria approximation algorithm for arbitrary cost,
%bounded length PARTIAL T-MULTI-PCUT in Section \ref{sect:bicriteria}.

\subsection{Approximations for T-MULTI-PCUT} \label{sect:covering_approx}
First, we present two approximation algorithms
for the constant $T$ problems T-PCUT and T-MULTI-PCUT, based upon 
Lemma \ref{lemm:cov_eq} and IP \ref{IP_mp}, when edge lengths have
a lower bound $q_{min} > 0$.
The idea
is as follows: 
for each path of vertices $p = v_1 \ldots v_l$  
between a pair of the target set
$\mathcal{S}$ with $d(p) = \sum_{i=1}^l d(v_{i-1}, v_i) \le T$, we must select at least one 
node belonging to the path into
the solution. Thus, we formulate the problem 
into a covering framework, where each node covers a subset of
paths. Both algorithms require the following
enumeration of paths.

\subsubsection{Path enumeration} \label{sect:path_enum}
This enumeration can be accomplished in polynomial-time in 
the following way: let
$T_0 = T / q_{min}$; then each path 
$p \in \mathcal{P} = \bigcup_{(u,v) \in \mathcal{S}} \mathcal{P}(u,v)$ must have at most
$T_0 + 1$ nodes. Thus, we may
iterate through all sequences of nodes of length at most $T_0$,
and test if the path produced is in $\mathcal{P}$; that is, for
some $(u,v) \in \mathcal{S}$,
the path must start at $u$, terminate at $v$, and satisfy $d(p) < T$. 
This procedure can be accomplished in time $O( n^{T_0} )$. 
Using these paths, we can construct the matrices $A^{(u,v)}$ in
IP \ref{IP_mp}.

\subsubsection{$O( \log n )$-approximation}
The first approximation algorithm for MULTI-PCUT
is given in Alg. \ref{alg:multipcut}. The general
approach is as follows. After the enumeration of all
paths in $\mathcal{P}$, the algorithm greedily selects
the node that intersects the largest number of paths
normalized by the vertex cost
until all paths in $\mathcal{P}$ have been covered.
By the proof of Lemma \ref{lemm:cov_eq}, when all such paths in 
$\mathcal{P}$ are covered, we have a feasible solution $W$.

An explicit description of the algorithm is given in
Alg. \ref{alg:multipcut}.
In lines 1 -- 3, the enumeration described above is performed.
Next, the algorithm initializes $W$, the set of vertices chosen,
and $C$, the set of paths covered by $W$ to $\emptyset$ in
line 4. The while loop on line 5 tests whether any paths
satisfying $d(p) \le T$ still exist in the network.
If so, it chooses the node $i^*$ which covers the most
such extant paths into the set $W$ on line 11 and updates
$C$ accordingly on line 12.

\begin{algorithm}[ht]
 \KwIn{Instance $(G, \mathcal{S},c,d)$ of T-MULTI-PCUT.}
 \KwOut{Subset $W \subset V$}
 \ForEach{$(u,v) \in \mathcal{S}$}{ 
   Compute $\mathcal{P}(u,v)$ explictly 
   as described in Section \ref{sect:path_enum}.
 }
 
 $C = \emptyset$,
 $W = \emptyset$\;
 \While{$\exists \, (u,v)$, $\mathcal{P}(u,v) - C \neq \emptyset$}{
   \For{$i \in V$}{
     $P_i = \{ p \in \mathcal{P} : i \in p$ and $p \cap W = \emptyset \}$ \;
     $b_i = |P_i| / c_i$\;
   }

   $i^* = \argmax \{ b_i \}$\;
   $W = W \cup \{i^* \}$\;
   $C = C \cup P_i$ \;
 }

 \caption{GEN: A Greedy, ENumerative, $O(\log n)$-approximation algorithm for MULTI-PCUT}
 \label{alg:multipcut} 
\end{algorithm}

\begin{theorem}
  Alg. \ref{alg:multipcut} achieves a performance guarantee
  of $O(\log n)$ with respect to the optimal solution
  with running time bounded by $O( k n^{T_0} )$.
  Furthermore, for each $n$, there exists an instance
  of the single pair PCUT problem where
  Alg. \ref{alg:multipcut} returns a solution of 
  cost greater than a factor 
  $\Omega ( \log n )$ of the optimal. \label{thm:GEN}
\end{theorem}
\begin{proof}
  The performance ratio of $O( \log n )$ follows from
  the fact that IP \ref{IP_mp} is a covering integer
  program corresponding to the set cover problem
  with at most $O( n^{T_0 + 1} )$ elements (the paths)
  for which the greedy algorithm has the
  ratio $O( T_0 \log n )$ \cite{Vazirani2001}. 

  Next, we construct a tight example 
  for Alg. \ref{alg:multipcut}; which holds
  even in the case of the single pair T-PCUT,
  for $T = 5$.
  At the beginning of the construction, $G$
  contains two isolated nodes, $s, t$. Add nodes
  $g_1, \ldots, g_k$ and edges $(s, g_i)$ for each
  $g_i$. Next, add nodes $o_1, o_2$ to the graph,
  along with edges $(o_1, t), (o_2, t)$. Then,
  for each $g_i$, add $2^{i - 1}$ disjoint paths
  of length 2 between $g_i$ and $o_1$, and similar
  paths between $g_i$ and $o_2$. Let $d(u,v) = 1$
  for all edges in $G$. For $k = 3$, see
  Fig. \ref{fig:tightness} in the Appendix for a depiction of
  the construction. 
  Then Alg. \ref{alg:multipcut} will select nodes
  $g_k, \ldots, g_1$ in that order, while the optimal
  solution is $\{ o_1, o_2 \}$.
\end{proof}
\subsubsection{$(T_0 + 1)$-approximation} 
Next, we present FEN in Alg. \ref{alg:FEN}, 
a frequency-based rounding 
algorithm for LP \ref{IP_mp}.
FEN first enumerates $\mathcal{P}$
and constructs LP \ref{IP_mp}. In this covering
program, each path intersects at most $T_0 + 1$ nodes,
as discussed above. Hence, the algorithm nexts 
solves LP \ref{IP_mp} to obtain optimal fractional solution
$\bar{w}$. Next, an integral solution $\hat{w}$ is obtained by rounding
\begin{equation} \label{eq:hatw}
\hat{w}_{i} = \begin{cases} 1 & \bar{w}_i \ge \frac{1}{T_0 + 1} \\ 0 & \text{otherwise} \end{cases}
\end{equation}
That $\hat{w}$ is a feasible solution follows from the fact that 
for each $(u,v) \in \mathcal{S}$ and $p \in \mathcal{P}(u,v)$,
constraint   
$\sum_{i=1}^n A_{p,i}^{(u,v)} \bar{w}_i \ge 1,$
so at least one $\bar{w}_i$ in the sum must satisfy $\bar{w}_i \ge 1 / (T_0 + 1)$,
since the sum has at most $T_0 + 1$ nonzero elements. Furthermore,
since the optimal fractional solution has cost at most
the cost of the optimal integral solution, and the cost of
$\hat{w}$ is within factor $T_0 + 1$ of $\bar{w}$, it follows
that FEN is an $( T_0 + 1 )$-approximation algorithm.
\begin{algorithm}[t]
 \KwIn{Instance $(G, \mathcal{S},c,d)$ of T-MULTI-PCUT.}
 \KwOut{Subset $W \subset V$}
 \ForEach{$(u,v) \in \mathcal{S}$}{ 
   Compute $\mathcal{P}(u,v)$ explictly 
   as described in Section \ref{sect:path_enum}.
 }
 
 Construct and solve LP \ref{IP_mp} to get fractional optimal solution $\bar{w}$\;
 Round $\bar{w}$ to $\hat{w}$ by Eq. \ref{eq:hatw}. Return $W = \hat{w}$\;
 \caption{FEN: A Frequency-based rounding, ENumerative, $(T_0 + 1)$-approximation 
   algorithm for MULTI-PCUT}
 \label{alg:FEN} 
\end{algorithm}

\subsection{Probabilistic approximation algorithm} \label{sect:gest}
In this section, we propose another approximation algorithm,
for T-PCUT and T-MULTI-PCUT when the length function is
bounded below. This algorithm, GEST, is intended
to more easily handle large values of $T_0$ than the
algorithms in the preceding section. The key for GEST is a procedure to efficiently estimate the number of paths between $(u,v)$ of length at most $T$ that each vertex $i \in V$ lies upon, which will guide the greedy selection of nodes. By theoretical analysis, we demonstrate that GEST is not only efficient, but also has a probabilistic performance guarantee. 

%Note: change wording and notation
\subsubsection{Algorithm overview and key results} 
The GEST algorithm is detailed in Alg. \ref{alg:sample_greedy}. 
As an overview, GEST 
iteratively selects nodes for removal based upon 
its estimation procedure, until
the distance between all pairs $(u,v)\in \mathcal{S}$ 
exceeds $T$.  
Define $\tau(S),\tau_{uv}(S)$ as the number of paths
in $\cup_{(u,v)\in S}\mathcal{P}(u,v)$, $\mathcal{P}(u,v)$ that $S$ intersects, respectively and $\sigma(S)$, $\sigma_{uv}(S)$ as
corresponding estimators. 
From the definition, we have $\tau(S)=\sum_{(u,v)\in \mathcal{S}}\tau_{uv}(S)$ and $\sigma(S)=\sum_{(u,v)\in \mathcal{S}}\sigma_{uv}(S)$. 
In each iteration of GEST, the node that maximizes $\sigma(W \cup \{i\}), i\in V\backslash W$ will be added to $W$, the set of selected nodes. 
The details of the estimator $\sigma(S)$ and the path sampling method are
discussed in Sections \ref{sect:estimators} and \ref{sect:hq}, respectively.

In the following, we will prove Theorem \ref{thm:gestapprox}, which establishes the key results on the probabilistic approximation ratio and time complexity of GEST.  
Before the proof, we introduce Lemma \ref{lemm:samples} on the number of samples $L$ for each pair to guarantee the accuracy of $\sigma(S)$. The proof of Lemma \ref{lemm:samples} is provided in Section
\ref{sect:sample_number}. The parameter $\alpha$ in $L$ can be used to balance running time and accuracy of the algorithm.

\begin{lemma}\label{lemm:samples}
Let the number of paths sampled for each $(u,v)\in\mathcal{S}$ be at 
least $L=3k^2\log ( 2 n^2 ) / 2\alpha^2$. 
Then, given a set $S\subset V$ and $\delta$ as the maximum degree in $G$, the inequality $|\tau(S)-\sigma(S)|< \alpha \delta^{T_0}$ holds with probability at least $1-1/n^3$. 
\end{lemma}

\begin{theorem}\label{thm:gestapprox}
  Given an instance $(G, c, d, \mathcal{S})$ of
  uniform vertex cost T-MULTI-PCUT whose length function $d$ is bounded below,
  let $\delta$ be the maximum degree in $G$.  
  With probability at least $1 - 1/n$, 
  Alg. \ref{alg:sample_greedy} returns a feasible solution $W$
  with cost within ratio
  $O \left( \alpha \delta^{T_0} + \log | \mathcal{S} | \right)$
  of optimal.
  The running time of Alg. \ref{alg:sample_greedy} is 
  $O( k^3n\log ( 2 n^2 ) / 2\alpha^2 )$.
\end{theorem}
\begin{proof}
  Let $\Delta_{x} \tau( S ) = \tau \left( S \cup \{ x \}  \right) - \tau (S), \forall S\subseteq V, \forall x\in V$; then for any $S \subset T$, observe that
  \begin{equation} \label{eq:tausm} \Delta_{x} \tau (S) \ge \Delta_{x} \tau (T). \end{equation}
We will apply Lemma \ref{lemm:samples} and consider that the inequality
therein always holds; later, we will consider the probability that
the inequality in Lemma \ref{lemm:samples} does not hold for 
some application.
Let $\varepsilon = 4 \alpha \delta^{T_0}$ and apply
Lemma \ref{lemm:samples}. 
By \eqref{eq:tausm}, we have:
  \begin{equation} \label{eq:apxsm}
    \Delta_{x} \sigma (S) \ge \Delta_{x} \sigma (T) - \varepsilon .
  \end{equation}
  Observe that Alg. \ref{alg:sample_greedy} at each iteration picks
  $a_i$ such that $a_i = \argmax \Delta_{a_i} \sigma ( \{ a_1, \ldots, a_{i-1} \} )$.
  Let $A_i = \{a_1, \ldots, a_i \}$ be the
  choice of Alg. \ref{alg:sample_greedy} after $i$ iterations, and let $A_g$ be the 
  final solution
  returned by the algorithm. Let $o = OPT$ be the size of an optimal
  solution $C = \{c_1, \ldots, c_{o} \}$ satisfying
  $\sigma ( C ) \ge P$, where $P$ is the number of paths
  in $\mathcal{P}$; notice that $\sigma ( S ) \ge P$
  is determined in Alg. \ref{alg:sample_greedy} by testing
  if all pairs in $\mathcal{S}$ satisfy $d(s,t) > T$ after removal
  of $S$.
Then
\begin{align}
  P - \sigma( A_i ) &\le \sigma(A_i \cup C) - \sigma(A_i) \nonumber \\
  &= \sum_{j = 1}^{o} \Delta_{c_j} \sigma \left( A_i \cup \{ c_1, \ldots, c_{j-1} \} \right) \nonumber \\
  &\le \sum_{j = 1}^{o} \Delta_{c_j} \sigma ( A_i ) + o\varepsilon \qquad \text{(by Eq. \ref{eq:apxsm} )} \nonumber \\
  &\le o \cdot \left[ \sigma( A_{i + 1} ) - \sigma( A_i ) + \varepsilon \right]. \label{increase_bound}
\end{align}
Therefore, $P - \sigma( A_{i + 1} ) - \varepsilon \le \left( 1 - \frac{1}{o} \right) (P - \sigma( A_i ) ).$
Then
\begin{align}
  P - \sigma(A_i) &\le P \left( 1 - \frac{1}{o} \right)^i + \varepsilon \sum_{j=0}^{i - 1} \left( 1 - \frac{1}{o} \right)^j \nonumber \\
  &\le P \left( 1 - \frac{1}{o} \right)^i + \varepsilon o. \label{ineq_diff}
\end{align}
From here, there exists an $i$ such that the
following differences satisfy
\begin{align}
P - \sigma(A_{i}) &\ge o( 1 + \varepsilon) \text{, and} \label{ineq_diff2} \\
P - \sigma(A_{i + 1}) &< o( 1 + \varepsilon ).
\label{ineq_size_bd}
\end{align}
Thus, by inequalities (\ref{ineq_diff}) and
(\ref{ineq_diff2}),
$o \le P \exp \left( \frac{-i}{o} \right),$
and 
$ i \le o \log \left( \frac{ P}{ o } \right).$
By inequality
(\ref{ineq_size_bd}) and the assumption on the
termination of the algorithm, the greedy
algorithm adds at most
$o(1 + \varepsilon )$ more elements, so
$g \le i + o(1 + \varepsilon ) \le o \left( 1 + \varepsilon  + \log \left( \frac{P}{o} \right) \right).$ In Alg. \ref{alg:sample_greedy}, we require the guarantee from Lemma \ref{lemm:samples} for all nodes $i\in V\backslash W$ for all iterations, which can happen $n^2$ times in the worst case. Therefore, by union bound, the probability of having the desired approximation ratio is at least $1-1/n$.
  The running time follows from the choice of $L$. Alg. \ref{alg:sample_greedy} needs to sample $k$ sets of $L$ samples per iteration and in the worst case, there can be $n$ iterations. 
\end{proof}

\begin{algorithm}[ht]
 \KwIn{Instance $(G,c,d,\mathcal{S})$, 
   accuracy parameter $\alpha \in (0, 1)$}
 \KwOut{Critical set of vertices $W$}
 $W = \emptyset$, $L = 3k^2\log ( 2 n^2 ) / 2\alpha^2$\;
 \While{$\exists$ pair $(u,v) \in \mathcal{S}$ with $d(u,v) < T$}{
   $x_i = 0$ for all $i \in V \backslash W$\;
   \ForEach{$(u,v) \in \mathcal{S}$ with $d(u,v) < T$}{
     Sample $L$ paths $\{q_1, \ldots, q_L \}$ in $\mathcal{R}(u,v)$ with Alg. \ref{alg:sample}\;
     \ForEach{$i \in V \backslash W$}{
       Compute estimator $\sigma_{uv}(W\cup \{i\})$ in Eq. \eqref{eq:est} using $\{q_1,q_2,...,q_L\}$\;
       $x_i = x_i + \sigma_{uv}(W\cup \{i\})$\;
     }
   }
   Let $i' = \argmax_i x_i$\;
   $W = W \cup \{ i' \}$\;
 }
 Return $W$\;
 \caption{GEST: A greedy estimation algorithm for MULTI-PCUT}
 \label{alg:sample_greedy} 
\end{algorithm}

\subsubsection{The estimators} \label{sect:estimators}
Let $u, v \in V$, and let $\mathcal{P}^i(u,v)$ be the
set of all paths $p$ between $u, v$ satisfying
the distance constraint $d(p) \le T$ and additionally
vertex $i \in p$. We want to 
efficiently estimate the 
quantity $\tau_{uv}(W\cup\{i\}) := | \cup_{j\in W\cup \{i\}}\mathcal{P}^j(u,v) |$ for all
$i \in V\backslash W$. 
To achieve this estimation, we
adapt the approach of Roberts et al. \cite{Roberts2007};
their estimators are for the total number
of simple paths in a graph, while we require
as estimation of the number of simple paths each 
vertex $v \in G$ lies upon, where the length
of each path is restricted to be at most $T$.

To define an estimator $\sigma_{uv}(W\cup\{i\})$, 
we proceed in the following way. 
Let $q$ be any simple path between
$u$ and $v$; we will define a probability
distribution $h(q)$ on paths $q$ satisfying
$h( q ) \neq 0$ if $q \in \mathcal{P}(u,v)$;
the\ distribution $h(q)$ is defined in 
Section \ref{sect:hq} and will have
domain $\mathcal{R}(u, v)$, a set of
simple paths starting from $u$.
We will then independently sample paths 
$q_1, \ldots, q_L$ from $h(q)$ and define
the estimator
\begin{equation} \label{eq:est}
  \sigma_{uv}(W\cup\{i\}) = \frac{1}{L} \sum_{l = 1}^L \frac{ I \left( q_l \in \cup_{j\in W\cup\{i\}}\mathcal{P}^j(u,v) \right)}{ h(q_l ) }, 
\end{equation}
where $I \left( q_l \in \cup_{j\in W\cup\{i\}}\mathcal{P}^j(u,v) \right)$ is an indicator random variable
that takes value $1$ if $W\cup\{i\} \cap q_j\neq \emptyset$ and $q_j \in \mathcal{P}(u,v)$, and
$0$ otherwise.
\begin{lemma}
  $\sigma_{uv}(W\cup\{i\})$ is an unbiased estimator of $\tau_{uv}(W\cup\{i\})$. 
\end{lemma}
\begin{proof}
  Let $Y(q)$ be the random variable
  \begin{equation} \label{eq:Yq}
    Y(q) := \frac{ I \left( q_l \in \cup_{j\in W\cup\{i\}}\mathcal{P}^j(u,v) \right) }{h(q)}, 
    \end{equation}
  for $q \in \mathcal{R}(u,v)$.
  Then the expection of $Y(q)$ is
  \begin{align*}
    \mathbf{E} \left( Y(q) \right)  %% &= \sum_{q \in \mathcal{R}(u,v)} Y(q) h(q) \\
    &= \sum_{q \in \mathcal{R}(u,v)} I \left( q \in \cup_{j\in W\cup\{i\}}P^j (u,v) \right) \\
    &= \left| \cup_{j\in W\cup\{i\}}P^j (u,v) \right| = \tau_{uv}(W\cup\{i\}).
  \end{align*}
  From here, the lemma follows from the law of large numbers.
\end{proof}
\subsubsection{Definition of $h(q)$ and path sampling} \label{sect:hq}
Next, we define the probability distribution $h(u)$ on $\mathcal{R}(u, v)$,
the set of all simple paths $q = u_0u_1 \ldots u_l$ 
starting from $u$ and ending at $v$
or ending at another vertex $v'$ and is maximal; that is, adding
any vertex $u_{l+1}$ to $q$ creates a cycle or causes the length
of the path to exceed $T$. We define the probability of a path
$q \in \mathcal{R}(u,v)$ sequentially: $h(q) := \prod_{i = 1}^l h( u_i | u_0u_1 \ldots u_{i - 1} ).$
Notice that $h(u_0) = h(s) = 1$ since $s$ is always chosen as
the starting vertex. Furthermore, $h( u_i | u_0 \ldots u_{i - 1} )$
is a uniform distribution over the number of vertices available
to be chosen as the next vertex of the path; that is $u_i$ does
not create a cycle and $d( u_0 \ldots u_i ) \le T$.

\begin{algorithm}[h]
 \KwIn{Graph $G$, pair of vertices $(u,v)$, $T$}
 \KwOut{A path $q \in \mathcal{R}(u,v)$, and probability value $h(q)$}
 $u_0 = u$, 
 $h = 1$,
 $i = 0$\;
 \While{$u_i \neq v$}{
   Set $N( u_i )$ equal to those 
   neighbors of $u_i$ not already in $q$, and whose addition
   to $q$ maintain $d(q) \le T$\;
   \If{ $N( u_i ) == \emptyset$ }{ break \; }
     
   Choose $u_{i + 1}$ from $N( u_i )$ with probability $1 / \left| N( u_i ) \right|$\;
   $h = h \cdot \frac{1}{ \left| N( u_i ) \right| }$, $q = q u_{i+1}$ \;
   $i = i + 1$ \;
 }

 Return $q = u_0, \ldots, u_l$, $h(q) = h$\;
 \caption{Algorithm for sampling from $\mathcal{R}(u,v)$}
 \label{alg:sample} 
\end{algorithm}
The definition of $h$ lends itself to the following sequential
sampling algorithm, shown in Alg. \ref{alg:sample}.
In line 1, the algorithm choose $u_0 = u$ with probability
$h = 1$. Let $N(u_i)$ be the set of neighbors of $u_i$
not previously chosen into the path $q$.
If $N( u_i ) = \emptyset$ or $u_i = v$, the algorithm
terminates. Otherwise $u_{i + 1}$ is chosen from
$N(u_i)$ uniformly with probability $1 / |N(u_i)|$
and the value of $h$ is updated accordingly.
\subsubsection{Bound on number of samples required} \label{sect:sample_number}
In this section, we prove Lemma
\ref{lemm:samples} for how many path samples are
required to ensure $|\tau_{uv}(S) - \sigma_{uv}(S)| \le \alpha \delta^{T_0}/k$.
To this end, we require Hoeffding's inequality 
\begin{theorem_nn}[Hoeffding's inequality]
  Suppose $Y_1, \ldots, Y_L$ are
  independent random variables in
  $[0, K]$. Let $Y = \frac{1}{L} \sum_{i=1}^L Y_i$.
  Then the probability
  $\mathbf{P} \left( \left|Y - \mathbf{E}(Y) \right| \ge t \right) \le 2 \exp \left( \frac{ -2 L t^2 }{ K^2 } \right).$
\end{theorem_nn}

\begin{proof}[Proof for Lemma \ref{lemm:samples}]
Consider $Y_i = Y(q_i)$, where
$Y(q)$ is the random variable defined in
(\ref{eq:Yq}).
Let $K \le \delta^{T_0}$, which is the maximum value of $Y_i$,
and $t = \alpha \delta^{T_0}/k$. 
Next, we require the probability bound from Hoeffding's inequality
to be less than $\frac{1}{n^3k}$. Solving for the
number of samples yields
$L \ge 3k^2\log ( 2 n^2 ) / 2\alpha^2.$
Therefore, when the number of samples is at least $L$, we can guarantee $|\tau_{uv}(S) - \sigma_{uv}(S)| \le \alpha \delta^{T_0}/k$ for one pair $(u,v)\in \mathcal{S}$ with probability $1-\frac{1}{n^3 k}$. Then, the inequality holds for all $(u,v)\in\mathcal{S}$ with probability $1-1/n^3$ by union bound. Since $\tau(S)$ and $\sigma(S)$ are the summations, $|\tau(S)-\sigma(S)|$ is at most $\alpha \delta^{T_0}$ when all the inequalities hold.
\end{proof}

\subsubsection{Further modification to GEST} \label{sect:gest-mod}
In this section, we discuss a simple modifications to
GEST; this modification, GESTA, improves performance for the T-MULTI-PCUT
problem.
%, and the second, GESTB, allows GEST to solve the
%PARTIAL T-MULTI-PCUT problem.

\emph{GESTA:} In practice, valid path samples in $\mathcal{P}$ 
become harder to obtain as GEST progresses nearer to a solution
to T-MULTI-PCUT; this fact results from most valid paths originally in the network 
having already been broken. Therefore, we propose
GESTA, a modification to Alg. \ref{alg:sample_greedy}
as follows: if GESTA performs $L$ samples, as in line 5 of GEST,
and obtains no valid paths
in $\mathcal{P}(u,v)$ for any $(u,v) \in \mathcal{S}$, then GESTA computes 
a shortest
path between a randomly chosen pair $(u,v)$ in $\mathcal{S}$ for which $d(u,v) \le T$.
The algorithm then chooses the cheapest node on this path into its solution, and continues
with the \texttt{while} loop on line 2 of GEST.

\section{Experimental evaluation} \label{sect:exp}
In this section, we experimentally evaluate our proposed algorithms
on the QoS vulnerability assessment TCVA in \ref{sect:tcva}. %proposed in Section \ref{sect:qos_vuln}.
In Section \ref{sect:methods}, we discuss the methodology of our evaluation.
%In Section \ref{sect:tcva}, we evaluate our algorithms for the targeted assessment
%TCVA, while the evaluation on the global assessment GCVA is presented in Section \ref{sect:gcva}.
\subsection{Datasets and methodology} \label{sect:methods}
\emph{Synthesized datasets:} To generate topologies,
we used a well-known Internet topology generator BRITE \cite{Medina2001};
which we employed to generate (1) Flat Router-Level (RL) only, (2) Flat 
Autonomous System level (AS) only,
and (3) hierarchical top-down datasets, consisting of AS and RL, with each
AS divided into routers.
We also used topologies generated according to Erdos-Renyi (ER) random graphs.
To simulate a QoS metric, edges were weighted uniformly in the interval $[1,10]$, following
\cite{Xue2007,Xuan2010}. The dataset statistics are as follows:
ER1, an ER graph with $n = 1000$, $m = 49995$; RL1, router-level 
graph with $n = 5000$, $m = 250000$, generated by BRITE with default parameters
and Waxman model;
RL2, same as RL1 except $n = 1000, m = 2000$; RL3, same as RL1 except
$n = 100$, $m = 200$; AS1, an AS-level graph generated by BRITE with
default parameters and $n = 10000, m = 498725$; and finally,
H1, a hierarchical BRITE top-down graph with 200 autonomous systems
and 100 routers per AS, with $n = 20000, m = 660604$.

% \emph{Real dataset:} The Intel Lab Data (ILD) \cite{Bodik2004} consists of the locations
% of 54 wireless sensors. Barsocchi et al. \cite{Barsocchi2007} developed a 
% frame error model based upon the distance in meters of two communicating
% wireless nodes. We used this model to generate a packet error rate between
% each pair of wireless sensors in ILD; the exact parameter values are given
% in Appendix \ref{apx:per}; to obtain an additive metric from 
% the packet error rate, we used transformation (\ref{eq:per}) 
% and Lemma \ref{lem:per} from Section \ref{sect:IIoT}.

\emph{Algorithms for TCVA:}
For TCVA, 
we compared the following algorithms with 
GEN (Alg. \ref{alg:multipcut}),
FEN (Alg. \ref{alg:FEN}),
and GESTA (Section \ref{sect:gest-mod}):
\begin{itemize}
 \item \emph{OPT:} the optimal solution of 
IP \ref{IP_mp}, which was implemented using the 
IP solver included in the open-source
GNU Linear Programming Kit (GLPK) \cite{Makhorin2012};
\item \emph{MC:} the classical
minimum-cut algorithm implemented with the Goldberg-Tarjan algorithm \cite{Goldberg1988}
for maximum flow, only employed when the size of the target set $|\mathcal{S}| = 1$; and
\end{itemize}

%This simple modification to GEST performed
%well in our experiments, allowing the algorithm to finish quickly even when 
%samples became difficult to obtain.
% \emph{Algorithms for GCVA:} We evaluated 
% ENBI (Alg. \ref{alg:beta}), with algorithm $\mathscr{A} = $ GEN and
% $\epsilon = 0.9$; and GESTB from Section \ref{sect:gest-mod}.
% We compared performance with the following algorithms:
% \begin{itemize}
%   \item \emph{OPT:} the optimal solution to IP \ref{IP_beta} implemented using GLPK;
%   \item \emph{MAXDEG:} chooses nodes
%     with highest weighted degree until $\beta$-fraction of pairs are
%     $T$-separated; and
%   \item \emph{BETWEEN:} chooses nodes
%     with highest betweenness centrality \cite{Vazirani2001} until $\beta$-fraction of pairs are
%     $T$-separated.
% \end{itemize}

%We also compared with 
%two baseline heuristics to choose critical nodes: MAXDEG, which ; and BETWEEN, which operates similarly to MAXDEG, except
%that it uses the betweenness centrality \cite{Vazirani2001} to pick nodes.

The cost function on vertices employed for TCVA is specified
in each section; when cost is uniform, we refer to the size of
the solution returned by each algorithm.
The path enumeration required for GEN, FEN, and OPT was
parallelized, using at most 25 threads. This parallelization was
accomplished by assigning distinct initial segments of paths to distinct threads.
Also, when $k > 1$, enumerations for distinct pairs were assigned to distinct
threads. Total computation time is the sum of the computation time over all threads. 
Algorithms were limited to one hour of wall-clock time before termination;
this could be much more computation time than one hour depending on the level
of parallelization. All times shown in the results are total computation time.
All experiments
were performed on a machine with Intel(R) Xeon(R) CPU E5-2697 v4 @ 2.30GHz
and 392 GB RAM. 

\subsection{Evaluation for Targeted Assessment (TCVA)} \label{sect:tcva}
\subsubsection{On choice of target set}
In order to evaluate the algorithms for TCVA, it is
necessary to choose the target set $\mathcal{S}$; in practice,
this choice is entirely up to the user. 
First, we discuss the motivation and effectiveness
of choosing the target sets $\mathcal{S}$ uniformly randomly; next,
we observe how restricting the elements of the target set based upon 
their degree affects the size of the optimal solution. 

\emph{Uniformly random:}
One method of evaluating the performance of our algorithms
for TCVA is to measure the average size (or cost) of
the solution over all possible choices of the
target set $\mathcal{S}$. 
To avoid the large computation time 
involved in running each algorithm on
each possible choice of $\mathcal{S}$,
we approximated this value by averaging over $N$ uniformly
random choices of $\mathcal{S}$. To justify this approximation,
we show in Fig. \ref{fig:vsN} the average cost of the solution 
returned by each algorithm versus $N$ on the RL1 dataset,
with $k = |\mathcal{S}| = 1000$ and $T=4$. Also shown is the sample standard
deviation of the $N$ values for the cost. While the value
of the mean fluctuates, the value of these fluctuations is less than 
$10 \%$ despite the huge number ${5000 \choose 1000}$, the number
of possible choices of $\mathcal{S}$.
Qualitatively similar results
were found for the other datasets and $k$ values. Therefore, in the remainder
of this section we average results over $N = 10$ uniformly 
random choices of $\mathcal{S}$
unless otherwise stated, which we found sufficient to identify
trends in the results.

\begin{figure}[ht]
  \centering
  \subfigure[RL1,$k=1000,T=4$] {
    \includegraphics[width=0.22\textwidth]{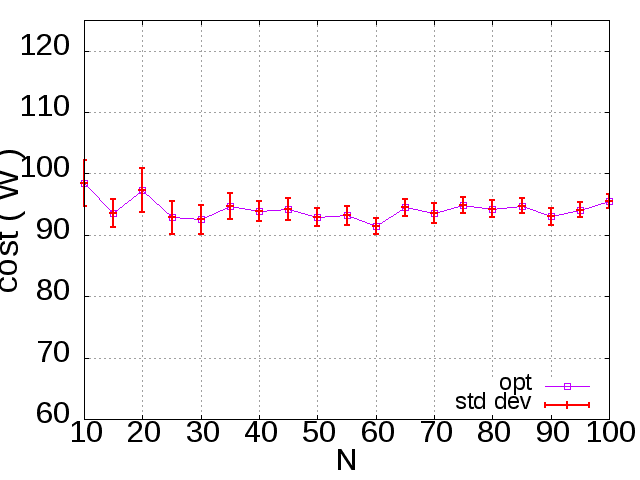}
    \label{fig:vsN}
  }
  \subfigure[RL1, $k=100, T = 5$] {
    \includegraphics[width=0.22\textwidth]{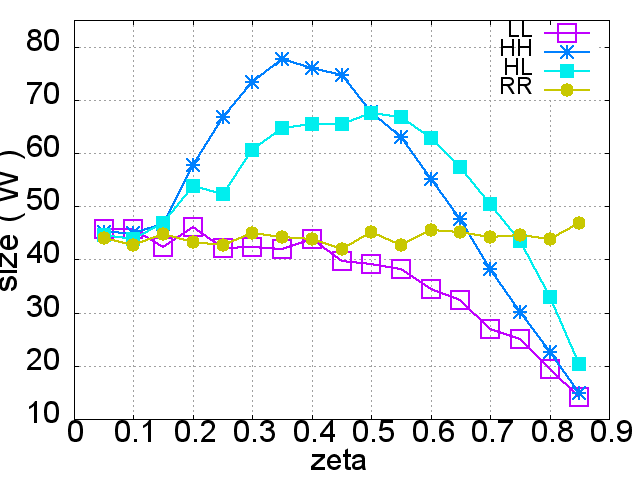}
    \label{fig:targs}
  }

\caption{(a): Average and standard deviation of cost values versus $N$, the number
of random choices of $\mathcal{S}$. (b): Impact of restricting the
choice of target set by degree on the size of the optimal solution to TCVA.} \label{fig:N}
\end{figure}

\emph{By degree:} Next, we observed how restricting the choice of
the target set by degree impacts the size of the optimal
solution. For the purposes of this assessment, let 
$\zeta \in (0,1)$, and let $\delta$ be the maximum
degree in graph $G = (V,E)$;
define the following two sets of vertices:
$H = \{ v \in V: d(v) \ge \zeta \delta \}$,
$L = \{ v \in V: d(v) \le (1 -\zeta) \delta \}$.
% \begin{align*} 
%   H &= \{ v \in V: d(v) \ge \zeta \delta \} \\
%   L &= \{ v \in V: d(v) \le (1 -\zeta) \delta \}.
% \end{align*}
Then we may restrict a source or target node to lie uniformly 
randomly within one of these sets. 
We consider four different schemes of choosing the target
set based upon $H,L$: HL, HH, LL, and RR. In HL,
for each pair $(s,t) \in \mathcal{S}$, $s$ is chosen uniformly
random from $H$, and $t$ is chosen uniformly randomly within $L$.
HH and LL are defined analogously, and RR chooses both nodes
of each pair uniformly randomly from the entire vertex, as 
in the previous section.

In Fig. \ref{fig:targs}, we plot the size of the optimal solution
to TCVA versus $\zeta$ for each scheme of target set selection,
averaged over $N=10$ choices of $\mathcal{S}$.
The results for LL and RR are as expected; RR shows no dependence
on $\zeta$, and LL is approximately equal to RR for low values
of $\zeta$ before decreasing monotonically as $\zeta$ approaches
1. However, HH and HL initially increase before decreasing below
RR -- this behavior is explained by the cardinality of $H$ and $L$
in addition to the restriction upon the degree. As $\zeta$ increases,
the cardinality of $H, L$ decrease; as these cardinalities decrease,
it becomes more likely that an element from one pair in the target
set appears in another pair, even though all pairs in the target set
$\mathcal{S}$ are distinct. As the fraction of nodes appearing in
multiple pairs increases, it becomes easier to pseudo-separate 
the target set. This effect counteracts the fact that higher degree
nodes are more difficult to pseudo-separate.

\begin{figure*}[ht]
  \centering
  \subfigure[RL1, $T = 6$] {
    \includegraphics[width=0.22\textwidth]{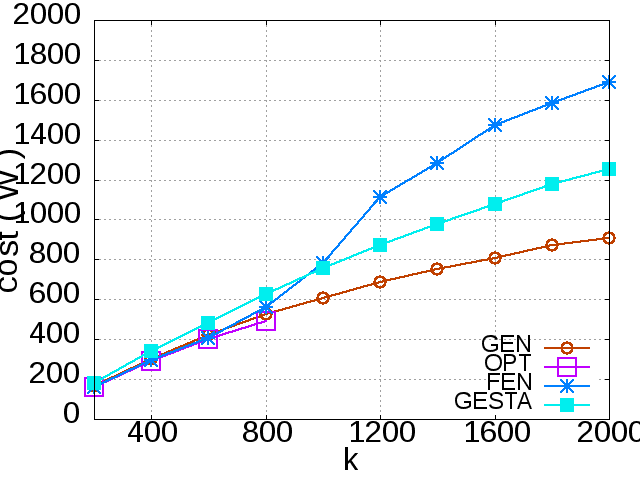}
    \label{fig:rl1_6}
  }
  \subfigure[AS1, $T = 6$] {
    \includegraphics[width=0.22\textwidth]{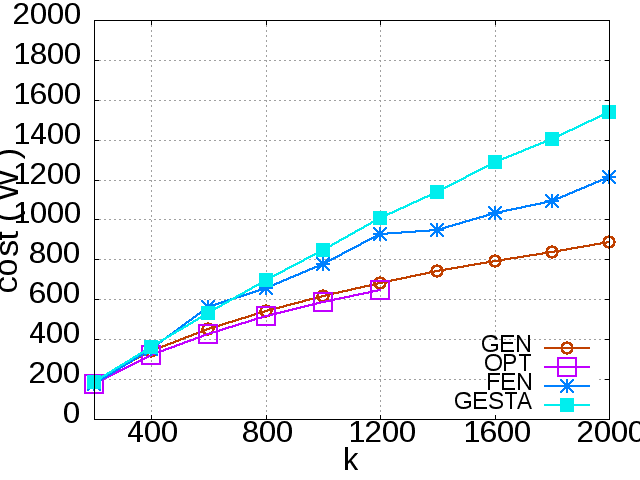}
    \label{fig:as1_6}
  }
  \subfigure[ER1, $k = 1$] {
    \includegraphics[width=0.22\textwidth]{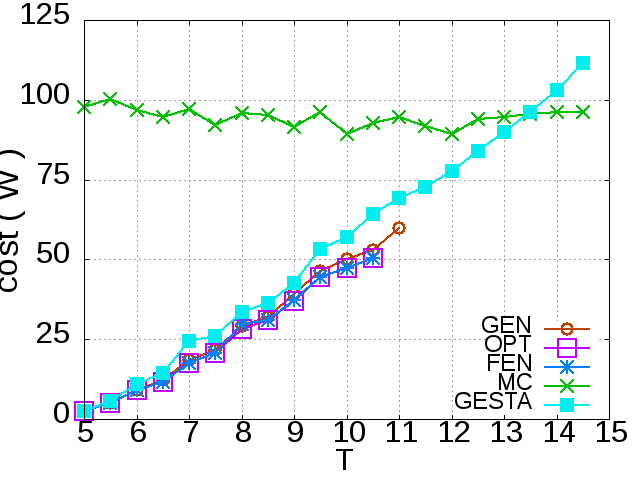}
    \label{fig:er1_1}
  }
  \subfigure[AS1, $k = 1000$] {
    \includegraphics[width=0.22\textwidth]{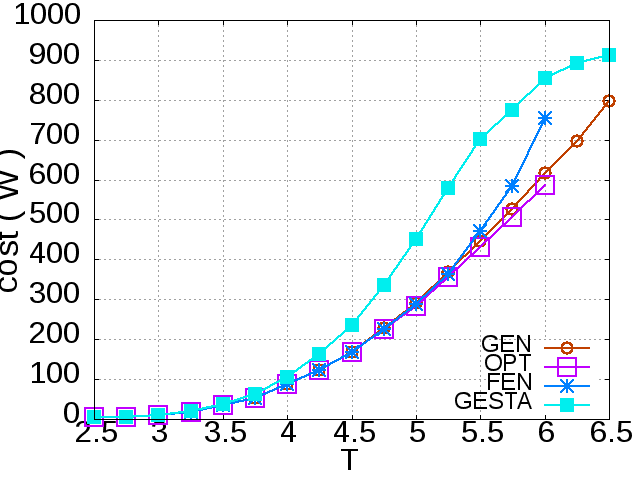}
    \label{fig:as1_5}
  }

  \subfigure[AS1, $T = 6$] {
    \includegraphics[width=0.22\textwidth]{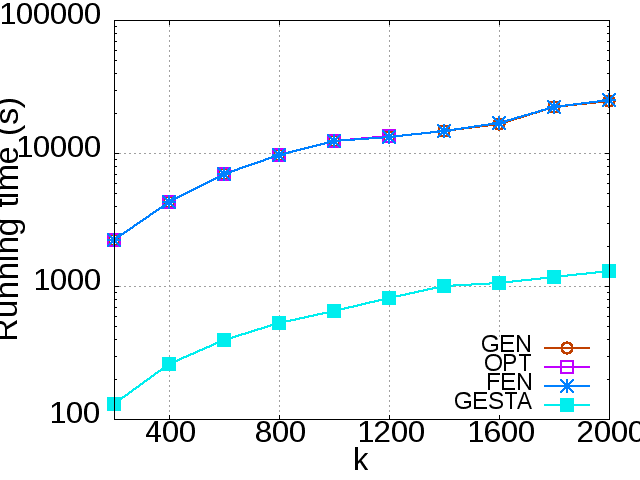}
    \label{fig:as1_6t}
  }
  \subfigure[H1, $T = 10$] {
    \includegraphics[width=0.22\textwidth]{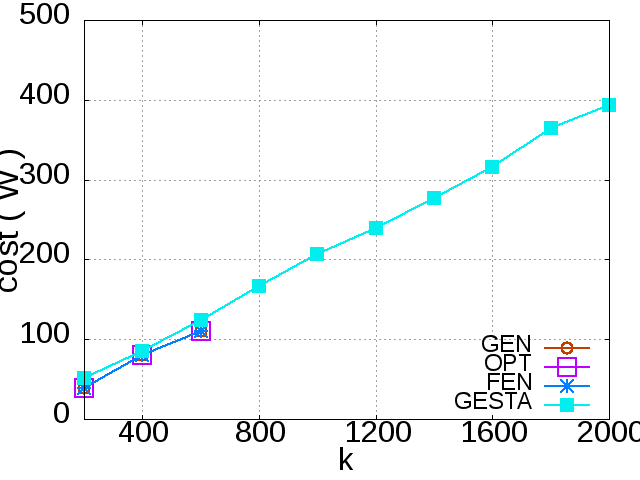}
    \label{fig:h1_6}
  }
  \subfigure[ER1, $k = 1000$] {
    \includegraphics[width=0.22\textwidth]{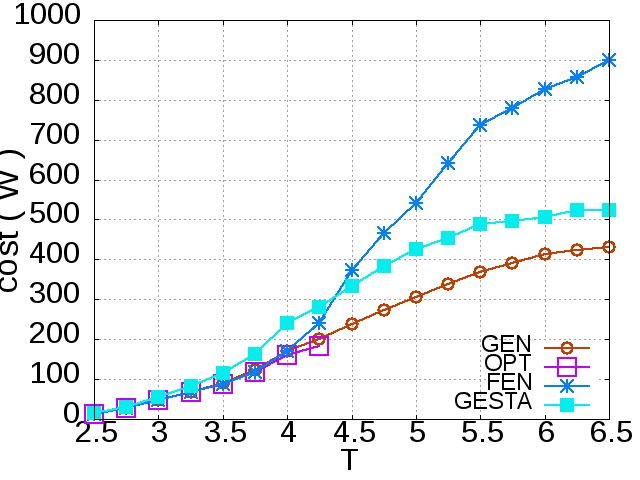}
    \label{fig:er1_5}
  }
  \subfigure[AS1, $k = 1000$] {
    \includegraphics[width=0.22\textwidth]{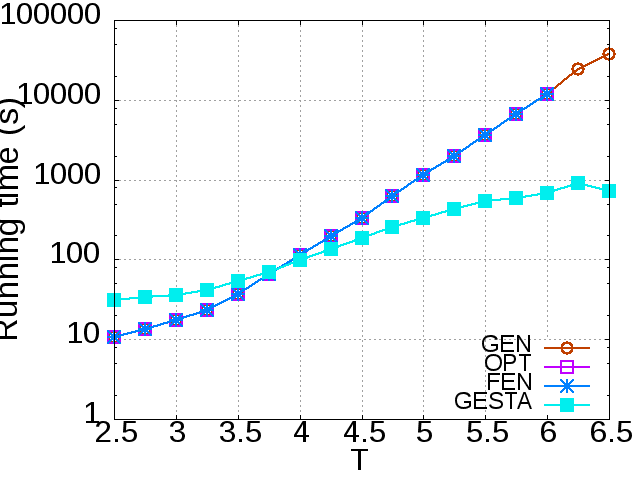}
    \label{fig:as1_5t}
  }

  % \subfigure[RL3, $T = 20$] {
  %   \includegraphics[width=0.22\textwidth]{results/r100_run11_C}
  %   \label{fig:rl3_g}
  % }
  % \subfigure[$(0.3, 20, c)$-critical vertices, RL3] {
  %   \includegraphics[width=0.22\textwidth]{results/rcomb.pdf} \label{fig:gephi}
  % }
  % \subfigure[RL2, $T = 25$] {
  %   \includegraphics[width=0.22\textwidth]{results/r1000_run11_C}
  %   \label{fig:rl2_g}
  % }
  % \subfigure[ILD, $\beta = 1$] {
  %   \includegraphics[width=0.22\textwidth]{results/new/ild_run1}
  %   \label{fig:ild1}
  % }

  \caption{(a) -- (h): Results for TCVA, described in the text.} \label{fig:kT}
%    (i): Cost of solutions versus $\beta$ for GCVA on dataset RL3.
%    (j): Visualization of $(0.3, 20, c)$-critical vertices chosen on RL3.
%    (k): Cost of solutions to GCVA versus $\beta$ on dataset RL2.
%    (l): Size of solution to GCVA versus packet error rate on dataset ILD.}
\end{figure*}
\subsubsection{Size of target set} \label{exp:target1} 
In this section, we fixed a constant $T$ for each dataset, let
vertices have uniform cost,
and observed the behavior of the algorithms when $k = |\mathcal{S}|$ was incremented
from $k = 200$ to $2000$. The only algorithm able to run 
on all datasets and $k$ values was GESTA, and it demonstrated
good performance (always within a factor of 2 in solution cost)
in comparison with OPT while running faster than
the other algorithms by a factor of more than 10.  
Representative results are shown in
the first two columns of Fig. \ref{fig:kT}. 
%First consider
%the average solution cost on datasets RL1, AS1, and H1 shown in
%Figs. \ref{fig:kT} (a), (b) and (f), respectively. 
GEN outperforms GESTA and is
the algorithm consistently the closest in performance to OPT
when both run. Second best alternates between GESTA and FEN on
RL1 and AS1, respectively.
For each dataset,
at some $k$ value, OPT exceeds one hour of computation time
and is no longer included in
the results. Notice on our largest dataset H1, with $T = 10$, 
neither GEN nor FEN can run after $k = 600$. Both of these
algorithms require the enumeration of $\mathcal{P}$, which was unable
to complete after this value of $k$ on this dataset. However,
on RL1 and AS1, GEN and FEN continue
to finish within one hour throughout the experiment;
notice from the running time shown in
Fig. \ref{fig:as1_6t} that the asymptotic behavior 
of the running time for fixed $T$
of GEN is linear in $k$, consistent with Theorem \ref{thm:GEN}.
In practice, GESTA runs faster than GEN and FEN by a constant factor of more than 10
on all inputs.

\subsubsection{Varying threshold $T$}
In this section, we consider two choices of $k$: $k = 1$, and
$k = 1000$. We then observed the behavior of the algorithms 
when $T$ was incremented; representative results are shown
in the last two columns of Fig. \ref{fig:kT}. When $k = 1$,
we compared the performance of our algorithms to the classical
MC algorithm (Fig. \ref{fig:er1_1}); as expected, MC returned a result independent of 
$T$, which demonstrates the inadequacy of solutions to the
classical cutting problems for
our assessments: for example, at $T = 7$,
MC is returning a solution of size more than four
times the optimal, and it does comparatively worse
for lower values of $T$.  
Also, we observe experimentally that as $T$ increases, we recover the
classical version of our problem: past $T = 13$, GESTA is completely separating the
input pair, and returning a solution of size similar to MC. 

As in the previous section, the only algorithm able to run for all
parameter values was GESTA, which maintained performance
within factor 2 of OPT. Although not as scalable as GESTA, GEN 
consistently outperformed the other algorithms in size of solution. 
On ER1, shown in Fig. \ref{fig:er1_1},
GEN was limited by the path enumeration time after $T = 11$, 
and FEN and OPT were unable to finish solving 
the LP \ref{IP_mp}; this LP solution is necessary for the rounding of FEN
and the integer solver of GLPK. Indeed,
the running time of GEN and FEN increased exponentially with $T$ (Fig. \ref{fig:as1_5t})
as expected. 

\subsubsection{Discussion}
Throughout the TCVA experiments, we consistently 
observed the best performance compared to the optimal 
by GEN, which was able to run in many situations where OPT could not finish. 
Furthermore, GEN scales well with the size of the target set $|\mathcal{S}|$.
However, as the threshold value $T$ becomes relatively large, 
LP \ref{IP_mp} becomes much larger and thus more difficult to solve; for this reason, GEN was 
unable to finish when $T$ became large.
In these cases, we demonstrated that the
approach of GESTA scales well with both the size of $|\mathcal{S}|$
and the threshold value $T$, while maintaining good performance with
respect to the optimal.

\section{Conclusions and Future Work} \label{sect:conclusion}
In this work, we introduced three new combinatorial pseudocut problems. We analyzed the
computational complexity of these problems, and we provided three approximation 
algorithms. We used the pseudocut problems to formulate a vulnerability
assessment TCVA with respect to an arbitrary additive QoS metric on a communications network.
Future work would include extending this assessment to incorporate 
more than one QoS metric; however,
this is likely to be difficult as the problem of finding a routing path satisfying two
or more QoS constraints is NP-hard; however, approximation algorithms do exist for this problem
\cite{Xue2007}. In addition, the computational complexity of the uniform
edge length version of our simplest problem, T-PCUT, is left open; our NP-hardness
proof required nonuniform edge lengths and we provided polynomial-time algorithms
only for special cases.

In our experimental evaluation, we found our $O( \log n )$-approximation GEN for T-MULTI-PCUT
to consistently return the solution closest to the optimal value, although its asymptotic
ratio is worse than the $(T + 1)$ ratio of FEN; however, for applications that demand a high
value for $T$, our experiments showed that GEN and FEN may be unsuitable, despite the
ease with which path enumeration may be parallelized -- for this case,
minor modifications to our probabilistic algorithm GEST were shown to give good performance in practice.
The modifications to GEST were necessary because of the difficulty of obtaining valid 
path samples when GEST is close to a feasible solution; future work would include boosting the
ability of GEST to obtain valid samples of paths between a terminal pair $(u,v) \in \mathcal{S}$,
so that heuristic modification GESTA becomes unnecessary.
%\clearpage
\bibliographystyle{unsrt}
\bibliography{info_16}

\appendix 
\subsection{Edge versions} \label{apx:edge}
Let $T$ be an arbitrary but fixed constant throughout this section.
The problems will take as input a triple $(G,c,d)$,
where $G$ is a directed graph $G = (V,E)$;
$c: E \to \mathbf{R}^+$ is a cost function on edges
representing the difficulty of removing
each edge; and $d: E \to \mathbf{R}^+$ is a length function on
edges.
Although both $c$ and $d$ may be considered weight
functions, we use \emph{cost} for $c$ and \emph{length} for $d$
to avoid confusion.
The distance $d(u, v)$ between two vertices is the
length of the $d$-weighted, 
directed, and shortest path between $u$ and $v$; the
cost $c(W)$ of set $W$ of a set of edges is the
sum of the costs of individual edges in $W$.
\begin{prob}[Minimum $T$-pseudocut (edge version)]
  Given triple $(G,c,d)$ and a pair $(s,t)$ of vertices of $G$, 
  determine a minimum cost set $W \subset E$ of
  edges such that
  $d( s, t ) > T$ after the removal of $W$ from $G$.
  \label{prob:min-pcut-edge}
\end{prob}

\begin{prob}[Minimum $T$-multi-pseudocut (edge version)]
  Given triple $(G,c,d)$, and a target set of pairs of vertices of $G$,
  $\mathcal{S} = \{ (s_1, t_1), (s_2, t_2), \ldots, (s_k, t_k) \}$,
  determine a minimum cost set $W$ of edges such that
  $d( s_i, t_i ) > T$ for all $i$ after the
  removal of $W$ from $G$. \label{prob:multi-pcut-edge}
\end{prob}

% \begin{prob}[Minimum partial $T$-multi-pseudocut (edge version)]
%   Given triple $(G,c,d)$, a target set of pairs of vertices of $G$,
%   $\mathcal{S} = \{ (s_1, t_1), (s_2, t_2), \ldots, (s_k, t_k) \}$, and
%   $\beta \in (0,1)$, 
%   determine a minimum cost set $W$ of edges such that
%   $d( s_i, t_i ) > T$ for at least $\beta k$ pairs $(s_i, t_i)$ after the
%   removal of $W$ from $G$. \label{prob:partial-multi-pcut-edge}
% \end{prob}

\subsection{Algorithms for edge versions} \label{apx:ev-algs}
If paths from $u$ to $v$ are defined as sequences of
edges instead of vertices, then, to approximate the 
edge versions, we can define analogous
approximation algorithms to  
GEN, FEN, GEST,
and ENBI
with analagous performance bounds.
For example, we 
define an analogous 
program to IP \ref{IP_mp} for the edge version of 
MULTI-PCUT below.

We will consider simple paths $p = p_0 p_1 \ldots p_l \in E$; that
is, paths containing no cycles. 
Let $\mathcal{P}( s_i , t_i )$ denote the set of simple paths $p$
between $(s_i,t_i) \in \mathcal{S}$ that satisfy the 
condition $d(p) \le T$. If an edge $u$ lies on path $p$,
we write $u \in p$. 
Consider the edge set of $G$ to be $\{1, \ldots, n \}$.
Let $A^{(u,v)}_{p, i} = 1$ if edge $i$ lies on path
$p \in \mathcal{P}(u, v)$, where $(u,v) \in \mathcal{S}$. 
If $i \not \in p$, let $A^{(u,v)}_{p, i} = 0$.
Also, let variable $w_i = 1$ if edge $i$ is to be
chosen into the set of edges $W$, 
and $0$ otherwise. Finally, denote the cost of choosing
edge $i$ as $c_i$, and let vectors $w = (w_1, \ldots, w_n)$ and
$c = (c_1, \ldots, c_n)$.
Then,
the covering $0-1$ integer program formulation 
is as follows.

\begin{IP}[Edge MULTI-PCUT] \label{IP_mp-edge}
\begin{align} 
  & \min c \cdot w \text{ such that } \nonumber \\
  &\sum_{i=1}^n A_{p,i}^{(u,v)} w_i \ge 1, \, \forall p \in P(u,v), \, \forall (u,v) \in \mathcal{S} \label{IP_cov-edge}\\
  & w_i \in \{ 0, 1 \}, \, \forall i \in \{1, \ldots, n \} .
\end{align}
\end{IP}

\end{document}